\documentclass[pre,twocolumn,10pt,aps,longbibliography,nofootinbib]{revtex4-1}

 \usepackage{verbatim}
 \usepackage{amsmath}
 \usepackage{amssymb}
 \usepackage{amsthm}
 \usepackage{latexsym}
 \usepackage{amsfonts}
 \usepackage{epsfig}
 \usepackage{epstopdf}
 \usepackage{wasysym} % for certain special symbols
 \usepackage{color}
 \definecolor{darkblue}{rgb}{0,0,.5}
 \usepackage[linktocpage, colorlinks=true, linkcolor=darkblue, citecolor=darkblue]{hyperref}
 \usepackage[all]{hypcap}

\newcommand{\C}[1]{{\cal{#1}}}
\newcommand{\bb}[1]{\textbf{#1}}

\newcommand{\bs}[1]{\boldsymbol{#1}}

\newcommand{\lr}[1]{{\left\langle {#1}\right\rangle}}

\begin{document}

\title{Non-Markovianity and negative entropy production rates}

\author{Philipp Strasberg}
%\email{philipp.strasberg@uni.lu}
\author{Massimiliano Esposito}
\affiliation{Physics and Materials Science Research unit, University of Luxembourg, L-1511 Luxembourg, Luxembourg}

\date{\today}

\begin{abstract}
 Entropy production plays a fundamental role in nonequilibrium thermodynamics to quantify the irreversibility of open 
 systems. Its positivity can be ensured for a wide class of setups, but the entropy production \emph{rate} can become 
 negative sometimes. This is often taken as an indicator of non-Markovianity. We make this link precise by showing 
 under which conditions a negative entropy production rate implies non-Markovianity and when it does not. 
 For a system coupled to a single heat bath this can be established within a unified language for two setups: 
 (i)~the dynamics resulting from a coarse-grained description of a Markovian master equation and 
 (ii)~the classical Hamiltonian dynamics of a system coupled to a bath. 
 The quantum version of the latter result is shown not to hold despite the fact that the integrated thermodynamic 
 description is formally equivalent to the classical case. The instantaneous fixed point of a non-Markovian dynamics 
 plays an important role in our study. Our key contribution is to provide a consistent theoretical framework to 
 study the finite-time thermodynamics of a large class of dynamics with a precise link to its non-Markovianity.
\end{abstract}

\maketitle
%\tableofcontents

\newtheorem{mydef}{Definition}[section]
\newtheorem{lemma}{Lemma}[section]
\newtheorem{thm}{Theorem}[section]
\newtheorem{crllr}{Corollary}[section]
\newtheorem*{thm*}{Theorem}%[section]
\theoremstyle{remark}
\newtheorem{rmrk}{Remark}[section]

%%%%%%%%%%%%%%%%%%%%%%%%%%%%%%%%%%%%%%%%%%%%%%%%%%%%%%%%%%%%%%%%%%%%%%%%%%%%%%%%%%%%%%%%%%%%%%%%%%%%%%%%%%%%%%%%%%%%%%%%
\section{Introduction}

The theory of stochastic processes provides a powerful tool to describe the dynamics of open systems. Physically, 
the noise to which these systems are subjected results from the fact that the system is coupled to an environment 
composed of many other degrees of freedom about which we have only limited information and control. This coarse-grained 
description of the system -- as opposed to the microscopic description involving the composite system \emph{and} 
environment -- is particularly appealing and tractable, when the Markovian approximation is applied. Therefore, Markovian 
stochastic dynamics are nowadays very commonly used to describe small open systems ranging from biochemistry (e.g., 
enzymes, molecular motors) to quantum systems (e.g., single atoms or molecules)~\cite{HillBook1977, SpohnRMP1980, 
VanKampenBook2007, BreuerPetruccioneBook2002}. Due to their outstanding importance for many branches of science, an 
entire branch of mathematics is also devoted to their study~\cite{KemenySnellBook1976}. 

A common feature of all Markovian processes is their \emph{contractivity}, i.e., the volume of accessible states shrinks 
monotonically during the evolution. This statement can be made mathematically precise by considering two arbitrary 
preparations, $p_\alpha(0)$ and $q_\alpha(0)$, describing different probabilities to find the system in state $\alpha$ 
at the initial time $t=0$. Their distance, as measured by the relative entropy 
$D[p_\alpha\|q_\alpha] \equiv \sum_\alpha p_\alpha \ln\frac{p_\alpha}{q_\alpha}$, 
monotonically decreases over time $t$, i.e., for all $t \ge 0$
\begin{equation}\label{eq contractivity}
 \frac{\partial}{\partial t} D[p_\alpha(t)\|q_\alpha(t)] \le 0.
\end{equation}
In other words, the ability to distinguish between any pair of initial states monotonically shrinks in time due to a 
continuous loss of information from the system to the environment. We note that also other distance quantifiers than the 
relative entropy fulfill Eq.~(\ref{eq contractivity}) and an analogue of Eq.~(\ref{eq contractivity}) also 
holds in the quantum regime where its violations has been proposed as an indicator of 
non-Markovianity~\cite{BreuerLainePiiloPRL2009, RivasHuelgaPlenioRPP2014, BreuerEtAlRMP2016}. 

The contractivity property~(\ref{eq contractivity}) of Markov processes gets another interesting physical 
interpretation in quantum and stochastic thermodynamics. In these fields, a nonequilibrium thermodynamics is 
systematically build on top of Markovian dynamics typically described by (quantum) master or Fokker-Planck 
equations~\cite{SchnakenbergRMP1976, JiangQianQianBook2004, EspositoHarbolaMukamelRMP2009, SekimotoBook2010, 
SeifertRPP2012, KosloffEntropy2013, SchallerBook2014, VandenBroeckEspositoPhysA2015}. In addition to being Markovian, 
the rates entering the dynamics must also satisfy local detailed balance. For a system coupled to a single heat bath, 
this ensures that the Gibbs state of the system is a null eigenvector of the generator of the dynamics at all times $t$. 
For autonomous dynamics, this implies that the fixed point of the dynamics is an equilibrium Gibbs state. 
For nonautonomous (also called \emph{driven}) dynamics, i.e., when some parameters are changed in time according to a 
prescribed protocol $\lambda_t$, the system in general does not reach a steady state, but the Gibbs state remains a null 
eigenvector of the generator of the dynamics at all times $t$. 
We call this an \emph{instantaneous} fixed point of the dynamics in the following. If we denote the Gibbs state of 
the system by $e^{-\beta E_\alpha(\lambda_t)}/\C Z(\lambda_t)$ with the energy $E_\alpha(\lambda_t)$ of state $\alpha$ 
and the equilibrium partition function $\C Z(\lambda_t) = \sum_\alpha e^{-\beta E_\alpha(\lambda_t)}$, the second law of 
thermodynamics for a driven system in contact with a single heat bath at inverse temperature $\beta$ can be expressed as 
\begin{equation}\label{eq 2nd law intro}
\dot\Sigma(t) = -\left.\frac{\partial}{\partial t}\right|_{\lambda_t} D\left[p_\alpha(t)\left\|\frac{e^{-\beta E_\alpha(\lambda_t)}}{\C Z(\lambda_t)}\right.\right] \ge 0.
\end{equation}
Here, the derivative is evaluated at fixed $\lambda_t$, i.e., $E_\alpha(\lambda_t)$ and $\C Z(\lambda_t)$ are treated as 
constants, which only depend parametrically on time. The quantity $\dot\Sigma(t)$ is the entropy production rate. 
Its positivity follows from the fact that the dynamics is Markovian \emph{and} that the Gibbs state is an instantaneous 
fixed point of the dynamical generator at all times. Within the conventional weak coupling and Markovian 
framework~\cite{SchnakenbergRMP1976, JiangQianQianBook2004, EspositoHarbolaMukamelRMP2009, SekimotoBook2010, 
SeifertRPP2012, KosloffEntropy2013, SchallerBook2014, VandenBroeckEspositoPhysA2015}, 
the entropy production rate can be rewritten as $\dot\Sigma(t) = \beta[\dot W(t) - d_t F(t)] \ge 0$, where $\dot W$ is 
the rate of work done on the system and $d_t F(t)$ denotes the change in non-equilibrium free energy (see 
Sec.~\ref{sec thermodynamics microlevel} for microscopic definitions of these quantities). The intimate connection between 
relative entropy and the second law was noticed some time ago in Ref.~\cite{ProcacciaLevineJCP1976} for undriven systems. 
In the undriven case, the precise form of Eq.~(\ref{eq 2nd law intro}) seems to appear first in Ref.~\cite{SpohnJMP1978} 
for quantum systems and it is discussed as a Lyapunov function in Ref.~\cite{VanKampenBook2007} for classical systems. 
The generalization to driven systems was given in Ref.~\cite{LindbladBook1983} and a similar form of 
Eq.~(\ref{eq 2nd law intro}) also holds for a system in contact with multiple heat 
baths~\cite{SpohnLebowitzAdvChemPhys1979}, see also Ref.~\cite{AltanerJPA2017} for a recent approach where 
Eq.~(\ref{eq 2nd law intro}) plays a decisive role. In this paper we will only focus on a single heat bath. 

While the Markovian assumption is widely used due to the enormous simplifications it enables, it is not always justified. 
Especially in stochastic thermodynamics an implicit but crucial assumption entering the Markovian description 
is that the degrees of freedom of the environment are always locally equilibrated with a well-defined associated 
temperature. This is in general only valid in the limit of time-scale separation where the environmental degrees of 
freedom can be adiabatically eliminated~\cite{EspositoPRE2012}. There is currently no consensus about the correct 
thermodynamic description of a system when the local equilibrium assumption for the environment is not met, i.e., when the 
system dynamics are non-Markovian. Especially, while different interesting results were obtained in 
Refs.~\cite{AndrieuxGaspardJSM2008, EspositoLindenbergPRE2008, RoldanParrondoPRL2010, RoldanParrondoPRE2012, 
LeggioEtAlPRE2013, BylickaEtAlSciRep2016} by \emph{starting} from a non-Markovian description of the system, the 
\emph{emergence} of non-Markovianity and its link to an underlying Markovian description of the microscopic degrees 
of freedom (system \emph{and} bath) was not yet established. 

The first main contribution of this paper is to provide a systematic framework for that situation able to investigate 
the influence of an environment, which is not locally equilibrated. While there has been recently great progress 
in the \emph{integrated} thermodynamic description of such systems~\cite{SeifertPRL2016, JarzynskiPRX2017, 
MillerAndersPRE2017, StrasbergEspositoPRE2017}, the instantaneous thermodynamic properties at the rate level were only 
studied in Ref.~\cite{StrasbergEspositoPRE2017}. We will here see that a remarkably similar framework to the conventional 
one above arises with the main difference that the entropy production rate $\dot\Sigma(t)$ can be negative sometimes. 
We then precisely link the occurence of $\dot\Sigma(t) < 0$ to underlying dynamical properties of the environment, thereby 
connecting the abstract mathematical property of (non-) Markovianity to an important physical observable. 

Our second main contribution is to establish a quantum counterpart for the classical strong coupling scenario studied by 
Seifert~\cite{SeifertPRL2016}. We find that the integrated thermodynamic description is very similar, but the 
instantaneous rate level description is \emph{not}. This hinders us to connect the occurence of negative entropy 
production rates to the non-Markovianity of the system evolution. We also provide an explicit example to show that recent 
claims in the literature about non-Markovianity, negative entropy production rates and steady states of dynamical maps do 
not hold. 

\emph{How to read this paper.---} This paper covers a wide range of applications from (i) rate master equations over (ii) 
classical Hamiltonian dynamics to (iii) quantum systems. We will keep this order in the narrative because it demonstrates 
beautifully the similarities and discrepancies of the different levels of description. We will start with a purely 
mathematical description of classical, non-Markovian systems, which arise from an arbitrary coarse-graining of an 
underlying Markovian network. While Sec.~\ref{sec lumpability Markov chains} reviews known results, 
Sec.~\ref{sec time dependent MEs} establishes new theorems (Appendices~\ref{sec app weak lumpability} 
and~\ref{sec app IFP} give additional technical details). Sec.~\ref{sec coarse-grained dissipative dynamics} can then 
be seen as a direct physical application of the previous section to the coarse-grained dynamics of a Markovian network 
obeying local detailed balance. In Sec.~\ref{sec classical system bath theory} we change the perpective and consider 
classical Hamiltonian system-bath dynamics, but with the help of Appendix~\ref{sec app Hamiltonian dynamics} we will see 
that we obtain identical results to Sec.~\ref{sec coarse-grained dissipative dynamics}. In our last general 
section~\ref{sec thermo quantum} we consider quantum systems. To illustrate the general theory, each subsection of 
Sec.~\ref{sec applications} is used to illustrate a particular feature of one of the previous sections. 
This roadmap of the paper is shown in Fig.~\ref{fig outline} and we wish to emphasize that it is also possible to 
read some sections independently. The paper closes by summarizing our results together with the state of the art of the 
field in Sec.~\ref{sec summary} and by discussing alternative approaches and open questions in Sec.~\ref{sec outlook}. 
We also provide an example to demonstrate that non-Markovian effects can speed up the erasure of a single bit of 
information, thereby showing that the field of non-Markovian finite-time thermodynamics provides a promising research 
direction for the future. 

\begin{figure}%[b]
 \centering\includegraphics[width=0.40\textwidth,clip=true]{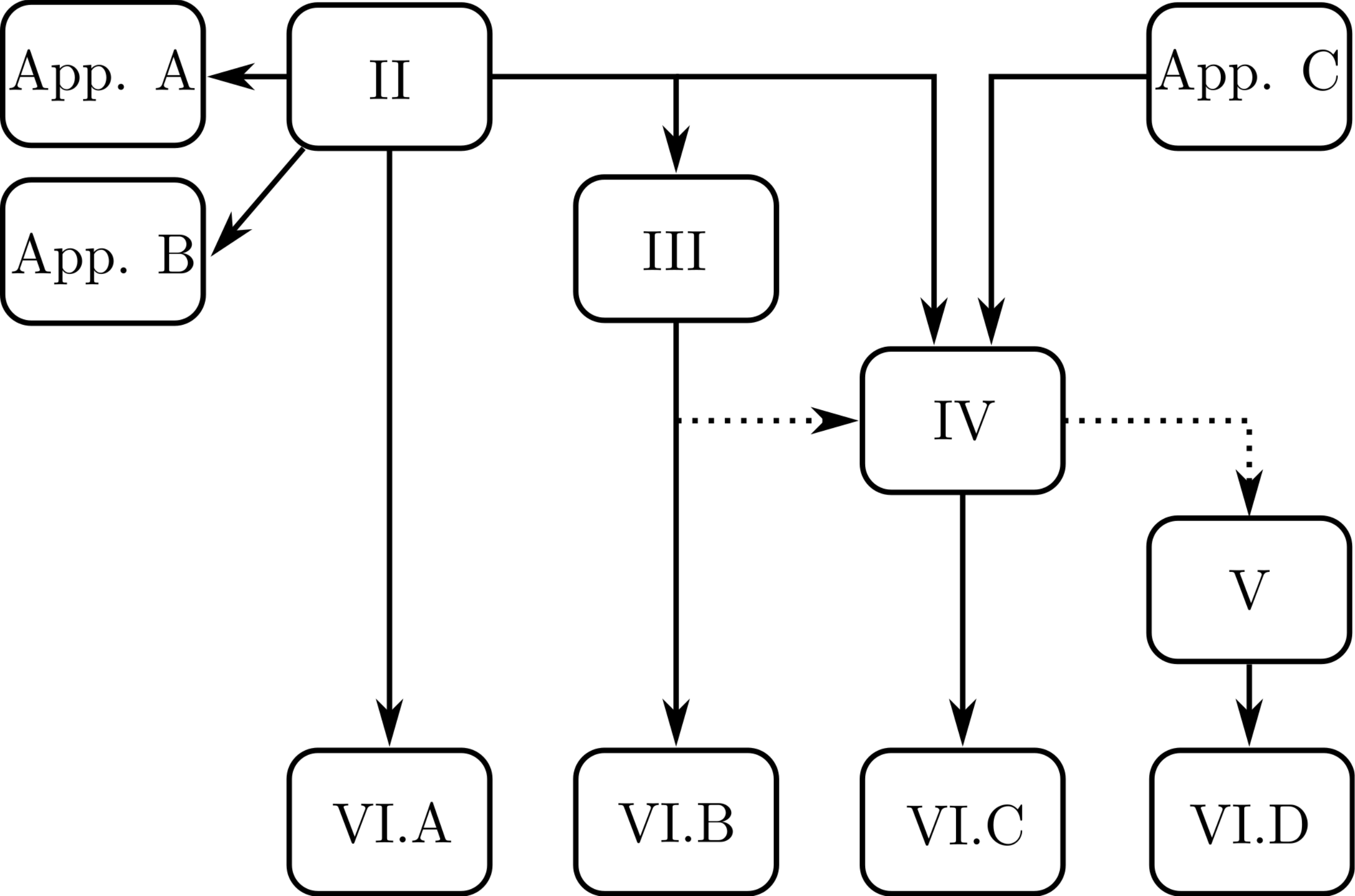}
 \label{fig outline} 
 \caption{``Roadmap'' of the paper with solid (dotted) arrows indicating strong (weak) dependencies. }
\end{figure}

The following abbreviations are used throughout the text: EP (entropy production), IFP (instantaneous fixed point), 
ME (master equation), TM (transition matrix), and TSS (time-scale separation).

%%%%%%%%%%%%%%%%%%%%%%%%%%%%%%%%%%%%%%%%%%%%%%%%%%%%%%%%%%%%%%%%%%%%%%%%%%%%%%%%%%%%%%%%%%%%%%%%%%%%%%%%%%%%%%%%%%%%%%%%
\section{Mathematical preliminaries}
\label{sec mathematical results}

%%%%%%%%%%%%%%%%%%%%%%%%%%%%%%%%%%%%%%%%%%%%%%%%%%%%%%%%%%%%%%%%%%%%%%%%%%%%%%%%%%%%%%%%%%%%%%%%%%%%%%%%%%%%%%%%%%%%%%%%
\subsection{Coarse-grained Markov chains}
\label{sec lumpability Markov chains}

In this section we establish notation and review some known results about Markov processes under coarse-graining. We will 
start with the description of a discrete, time-homogeneous Markov chain for simplicity, but soon we will move to the 
physically more relevant case of an arbitrary continuous-time Markov process described by a ME. Finally, we also introduce 
the concept of lumpability~\cite{KemenySnellBook1976}. 

\emph{Discrete, homogeneous Markov chains.---} We consider a Markov process on a discrete space $\C X$ with $N$ states 
$x\in\C X$ with a fixed TM $T_{\tau}(x|x')$, which propagates the state of the system such that 
\begin{equation}
 p_x(n\tau+\tau) = \sum_{y} T_{\tau}(x|y) p_{y}(n\tau) ~~~ (n\in\mathbb{N}),
\end{equation}
or in vector notation $\bb p(n\tau+\tau) = T_\tau \bb p(n\tau)$. Here, $p_x(n\tau)$ is the probability to find the 
system in the state $x$ at time $n\tau$, where $\tau > 0$ is an arbitrary but fixed time step (here and in what 
follows we will set the initial time to ${t_0 \equiv 0}$). Probability theory demands that $\sum_x p_x(n\tau) = 1$, 
$p_x(n\tau) \ge 0$ for all $x$, $\sum_x T_{\tau}(x|y) = 1$ and $T_{\tau}(x|y) \ge 0$ for all $x,y$. 
The steady state of the Markov chain is denoted by $\pi_x$ and it is defined via the equation 
$\boldsymbol\pi = T_{\tau}\boldsymbol\pi$. In this section we exclude the case of multiple steady states for 
definiteness, although large parts of the resulting theory can be applied to multiple steady states as 
well.\footnote{The contractivity property of Markov chains, Eqs.~(\ref{eq contractivity}) and~(\ref{eq 2nd law intro}), 
which plays an important role in the following, holds true irrespective of the number of steady states. } 

Next, we consider a partition $\boldsymbol\chi = \{\chi_1,\dots,\chi_M\}$ ($1<M<N$) of the state space such that 
\begin{equation}
 \bigcup_{\alpha=1}^M \chi_\alpha = \C X, ~~~ \chi_\alpha \cap\chi_\beta = \emptyset \text{ for } \alpha \neq \beta. 
\end{equation}
In the physics literature this is known as a coarse-graining procedure where different ``microstates'' $x$ are collected 
together into a ``mesostate'' $\alpha$, whereas in the mathematical literature this procedure is usually called lumping. 
In the following we will use both terminologies interchangeably and we denote a microstate $x$ belonging to the mesostate 
$\alpha$ by $x_\alpha$, i.e., $x_\alpha\in\chi_\alpha$. The idea is illustrated in Fig.~\ref{fig lumping}. 
We remark that tracing out the degrees of freedom of some irrelevant system (usually called the ``bath'') is a special 
form of coarse-graining. We will encounter this situation, e.g., in Sec.~\ref{sec classical system bath theory}. 

\begin{figure}%[b]
 \centering\includegraphics[width=0.20\textwidth,clip=true]{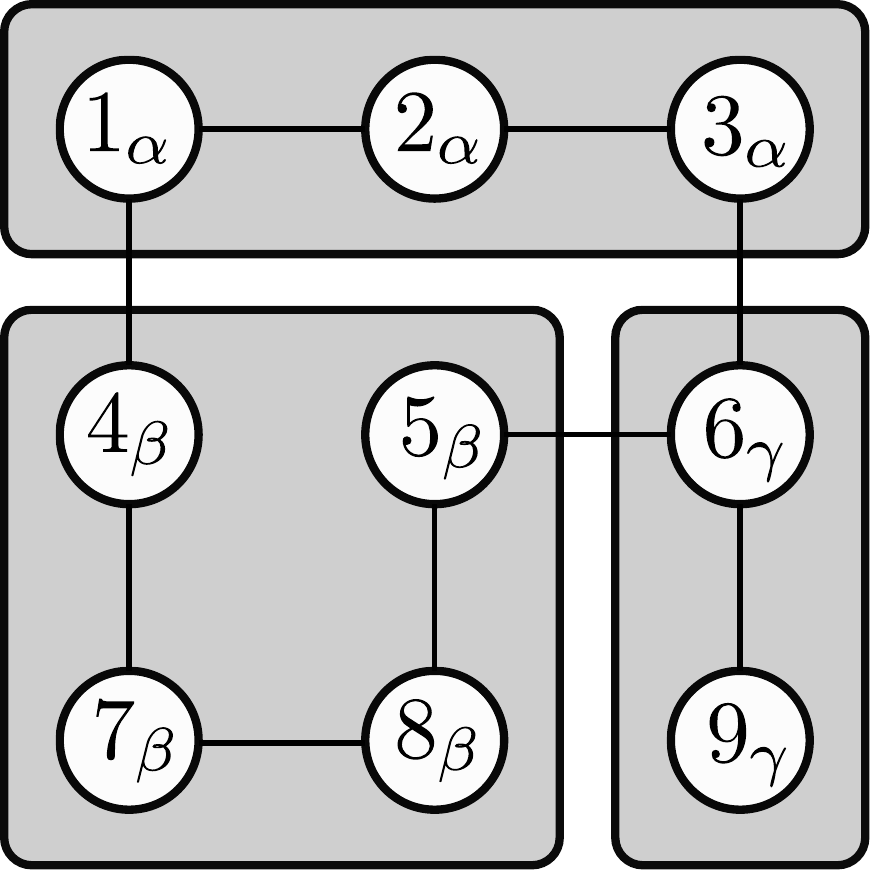}
 \label{fig lumping} 
 \caption{Lumping/coarse-graining of a discrete Markov chain with microstate-space $\C X = \{1,2,3,4,5,6,7,8,9\}$ into 
 three mesostates according to the partition $\boldsymbol\chi = \{\chi_\alpha,\chi_\beta,\chi_\gamma\}$ with 
 $\chi_\alpha = \{1,2,3\}$, $\chi_\beta = \{4,5,7,8\}$ and $\chi_\gamma = \{6,9\}$ (grey areas). Possible transistions 
 for which $T_{\tau}(x|y) \neq 0$ are depicted by a solid line connecting state $x$ and $y$. }
\end{figure}

Any partition $\boldsymbol\chi$ defines a stochastic process on the set of mesostates by considering for a given initial 
distribution $p_{x}(0)$ the probabilities to visit a sequence of mesostates $\alpha,\beta,\gamma,\dots$ at times 
$0,\tau,2\tau,\dots$ with joint probabilities 
\begin{equation}
 \begin{split}
  & p(\beta,\tau;\alpha,0) = \sum_{y_\beta,x_\alpha} T_{\tau}(y_\beta|x_\alpha)p_{x|\alpha}(0)p_\alpha(0),	\\
  & p(\gamma,2\tau;\beta,\tau;\alpha,0) =	\\
  & \sum_{z_\gamma,y_\beta,x_\alpha} T_{\tau}(z_\gamma|y_\beta) T_{\tau}(y_\beta|x_\alpha)p_{x|\alpha}(0)p_\alpha(0),
 \end{split}
\end{equation}
etc., where $p_\alpha(0) = \sum_{x_\alpha} p_{x_\alpha}(0)$ is the marginalized initial mesostate and 
$p_{x|\alpha}(0) = p_{x_\alpha}(0)/p_\alpha(0)$ is the initial microstate conditioned on a certain mesostate $\alpha$. 
The so generated hierarchy of joint probabilities $p(\alpha_n,n\tau;\dots;\alpha_1,\tau;\alpha_0,0)$ completely specifies 
the stochastic process at the mesolevel. It is called Markovian whenever the conditional probabilities 
\begin{equation}
 \begin{split}
  & p(\alpha_n,n\tau|\alpha_{n-1},n\tau-\tau;\dots;\alpha_0,0)	\\
  & \equiv \frac{p(\alpha_n,n\tau;\dots;\alpha_0,0)}{p(\alpha_{n-1},n\tau-\tau;\dots;\alpha_0,0)}
 \end{split}
\end{equation}
satisfy the Markov property~\cite{KemenySnellBook1976, VanKampenBook2007, RivasHuelgaPlenioRPP2014, BreuerEtAlRMP2016} 
\begin{equation}\label{eq cond Markovianity}
 \begin{split}
  & p(\alpha_n,n\tau|\alpha_{n-1},n\tau -\tau;\dots;\alpha_0,0)   \\
  & = p(\alpha_n,n\tau|\alpha_{n-1},n\tau-\tau). 
 \end{split}
\end{equation}
In practice this requires to check infinitely many conditions. But as we will see below, to compute all quantities of 
thermodynamic interest, only the knowledge about the evolution of the \emph{one}-time probabilities 
$p(\alpha_n,n\tau)$ is important for us. 

To see how non-Markovianity affects the evolution of the one time-probabilities, we introduce the following 
matrices derived from the above joint probabilities 
\begin{align}
 G_{\tau,0}(\beta|\alpha)	&=	\frac{p(\beta,\tau;\alpha,0)}{p_\alpha(0)} = \sum_{y_\beta,x_\alpha} T_{\tau}(y_\beta|x_\alpha)p_{x|\alpha}(0),	\label{eq TM mesolevel}	\\
 \tilde G_{2\tau,\tau}(\gamma|\beta)	&=	\frac{p(\gamma,2\tau;\beta,\tau)}{p_\beta(\tau)} = \frac{\sum_\alpha p(\gamma,2\tau;\beta,\tau;\alpha,0)}{\sum_\alpha p(\beta,\tau;\alpha,0)},	\nonumber	\\
 G_{2\tau,0}(\gamma|\alpha)	&=	\frac{p(\gamma,2\tau;\alpha,0)}{p_\alpha(0)}	\nonumber	\\
				&=	\sum_{z_\gamma,x_\alpha}\sum_{\beta,y_\beta} T_{\tau}(z_\gamma|y_\beta) T_{\tau}(y_\beta|x_\alpha)p_{x|\alpha}(0). \nonumber
\end{align}
Formally, these matrices are well-defined conditional probabilities because they are positive and normalized. 
However, we have deliberately choosen a different notation for $\tilde G_{2\tau,\tau}$ because only $G_{\tau,0}$ 
and $G_{2\tau,0}$ can be interpreted as \emph{transition} probabilities (or matrices) as they generate the 
correct time evolution for \emph{any} initial mesostate $p_\alpha(0)$. The matrix $\tilde G_{2\tau,\tau}$ instead 
depends on the specific choice of $p_\alpha(0)$: if we start with a different initial mesostate 
$q_\alpha(0)\neq p_\alpha(0)$, we cannot use $\tilde G_{2\tau,\tau}$ to propagate 
$q_\beta(\tau) = \sum_\beta G_{\tau,0}(\beta|\alpha) q_\alpha(0)$ further in time. This becomes manifest by realizing 
that the so generated hierarchy of conditional probabilities does not in general obey the Chapman-Kolmogorov equation, 
\begin{equation}
 G_{2\tau,0}(\gamma|\alpha) = \sum_\beta\tilde  G_{2\tau,\tau}(\gamma|\beta)G_{\tau,0}(\beta|\alpha).
\end{equation}

A way to avoid this undesired feature is to define the TM from time $\tau$ to $2\tau$ via the inverse of $G_{\tau,0}$ 
(provided it exists)~\cite{HaenggiThomasZPB1977, RivasHuelgaPlenioPRL2010, RivasHuelgaPlenioRPP2014, BreuerEtAlRMP2016} 
\begin{equation}\label{eq def intermediate TM}
 G_{2\tau,\tau} \equiv G_{2\tau,0} G_{\tau,0}^{-1}.
\end{equation}
The TM $G_{2\tau,\tau}$ does not depend on the initial mesostate, preserves the normalization of the state and by 
construction, it fulfills the Chapman-Kolmogorov equation: $G_{2\tau,0} = G_{2\tau,\tau}G_{\tau,0}$. However, 
as the inverse of a positive matrix is not necessarily positive, $G_{2\tau,\tau}$ can have negative entries. This clearly 
indicates that $G_{2\tau,\tau}(\gamma|\beta)$ cannot be interpreted as a conditional probability and hence, the process 
must be non-Markovian. Based on these insights we introduce a weaker notion of Markovianity, which we coin 
1-Markovianity. In the context of open quantum systems dynamics this notion is often simply called 
Markovianity~\cite{RivasHuelgaPlenioRPP2014, BreuerEtAlRMP2016}: 

\begin{mydef}[1-Markovianity]\label{def 1 Markovian}
 A stochastic process is said to be 1-Markovian, if the set of TMs $\{G_{n\tau,m\tau}|n\ge m\ge 0\}$ introduced above 
 fulfill $G_{n\tau,m\tau}(\alpha|\beta) \ge0$ for all $n\ge m\ge 0$ and all $\alpha,\beta$. 
\end{mydef}

It is important to realize that the notion of 1-Markovianity is weaker than the notion of Markovianity: 
if the coarse-grained process is Markovian, then it is also 1-Markovian and the TMs coincide with the conditional 
probabilities in Eq.~(\ref{eq cond Markovianity}). Furthermore, there exist processes which are 1-Markovian but not 
Markovian according to Eq.~(\ref{eq cond Markovianity}) (see, e.g., Ref.~\cite{RivasHuelgaPlenioRPP2014}). 

Before we consider MEs, we introduce some further notation. We let 
\begin{equation}
 \C A(0) \equiv \{p_x(0)|p_\alpha(0) \text{ arbitrary, } p_{x|\alpha}(0) \text{ fixed}\}
\end{equation}
be the set of all \emph{physically admissible initial states} with respect to a partition $\boldsymbol\chi$ (whose 
dependence is implicit in the notation). The reason to keep $p_{x|\alpha}(0)$ fixed is twofold: 
first, in an experiment one usually does not have detailed control over the microstates, and second, the 
TMs~(\ref{eq TM mesolevel}) for the lumped process depend on $p_{x|\alpha}(0)$, i.e., every choice of 
$p_{x|\alpha}(0)$ defines a different stochastic process at the mesolevel and should be treated separately. 
Which of the mesostates $p_\alpha(0)$ we can really prepare in an experiment is another interesting (but for us 
unimportant) question; sometimes this could be only a single state (e.g., the steady state $\pi_\alpha$). Of particular 
importance for the applications later on will be the set 
\begin{equation}\label{eq stationary preparation class}
 \C A_{\pi} \equiv \{p_\alpha\pi_{x|\alpha}|p_\alpha \text{ arbitrary}\}
\end{equation}
where $\pi_{x|\alpha} = \pi_{x_\alpha}/\pi_\alpha$ is the conditional steady state. Experimentally, such a class of 
states can be prepared by holding the mesostate fixed while allowing the microstates to reach steady state. Finally, 
we define the set of time-evolved admissible initial states 
\begin{equation}\label{eq time evolved admissible states}
 \C A(\tau) \equiv \{\bb p(\tau) = T_\tau\bb p(0)| \bb p(0) \in \C A(0)\}.
\end{equation}

\emph{Time-dependent MEs.---} For many physical applications it is indeed easier to derive a ME, which describes the 
continuous time evolution of the system state, compared to deriving a TM for a finite 
time-step~\cite{VanKampenBook2007, BreuerPetruccioneBook2002}. The ME reads in general 
\begin{equation}\label{eq ME general}
 \frac{\partial}{\partial t}p_{x}(t) = \sum_{y} W_{x,y}(\lambda_t) p_{y}(t)
\end{equation}
or in vector notation $\partial_t \bb p(t) = W(\lambda_t)\bb p(t)$. The rate matrix $W(\lambda_t)$ fulfills 
$\sum_{x} W_{x,y}(\lambda_t) = 0$ and $W_{x,y}(\lambda_t) \ge 0$ for $x\neq y$ and it is now also allowed to be 
parametrically dependent on time through a prescribed parameter $\lambda_t$. This situation usually arises by subjecting 
the system to an external drive, e.g., a time-dependent electric or magnetic field. Furthermore, we assume that the 
rate matrix has one IFP, which fulfills $W(\lambda_t) \boldsymbol\pi(\lambda_t) = 0$. 
Clearly, the steady state will in general also parametrically depend on $\lambda_t$. 

We can connect the ME description to the theory above by noting that the TM over any finite time interval ${[t,t+\tau]}$ 
is formally given by 
\begin{equation}\label{eq TM from ME}
 T_{t,t+\tau} = \C T_+ \exp\int_t^{t+\tau} W(\lambda_s) ds,
\end{equation}
where $\C T_+$ is the time-ordering operator. In particular, if we choose $\delta t = \tau/N$ small enough such that 
$\lambda_{t+\delta t} \approx \lambda_t$ (assuming that $\lambda_t$ changes continuously in time), we can approximate 
the TM to any desired accuracy via 
\begin{equation}
 T_{t+\tau,t} \approx \prod_{i=0}^{N-1} T_{t+i\delta t + \delta t,t+i\delta t} \equiv \prod_{i=0}^{N-1} e^{W(\lambda_{t+i\delta t}) \delta t}.
\end{equation}
As a notational convention, whenever the system is undriven (i.e., $\dot\lambda_t = 0$ for all $t$), we will simply drop 
the dependence on $\lambda_t$ in the notation. 

We now fix an arbitrary partition $\boldsymbol\chi$ as before. To describe the dynamics at the mesolevel, one can use 
several formally exact procedures, two of them we mention here. First, from Eq.~(\ref{eq ME general}) we get by direct 
coarse-graining 
\begin{equation}
 \begin{split}\label{eq effective rate matrix CG}
  \frac{\partial}{\partial t} p_\alpha(t)	&=	\sum_\beta R_{\alpha,\beta}[\lambda_t,p_\alpha(0)] p_\beta(t),	\\
  R_{\alpha,\beta}[\lambda_t,p_\alpha(0)]	&\equiv	\sum_{x_\alpha,y_\beta} W_{x_\alpha,y_\beta}(\lambda_t) p_{y|\beta}(t).
 \end{split}
\end{equation}
Here, the matrix $R[\lambda_t,p_\alpha(0)]$ still fulfills all properties of an ordinary rate matrix: 
$\sum_\alpha R_{\alpha,\beta}[\lambda_t,p_\alpha(0)] = 0$ and $R_{\alpha,\beta}[\lambda_t,p_\alpha(0)] \ge 0$ for 
$\alpha\neq\beta$. However, it explicitly depends on the initial mesostate $p_\alpha(0)$, which influences 
$p_{y|\beta}(t)$ for later times $t$. This is analogous to the problem mentioned below Eq.~(\ref{eq TM mesolevel}): 
the TMs computed with Eq.~(\ref{eq effective rate matrix CG}) at intermediate times depend on the initial state of the 
system. This reflects the non-Markovian character of the dynamics and makes it inconvenient for practical applications. 
Note that Eq.~(\ref{eq effective rate matrix CG}) still requires to solve for the full microdynamics and does not 
provide a closed reduced dynamical description. 

A strategy to avoid this undesired feature follows the logic of Eq.~(\ref{eq def intermediate TM}) and only makes use of the 
well-defined transition probability [cf.~Eq.~(\ref{eq TM mesolevel})] 
\begin{equation}\label{eq TM mesolevel general}
 G_{t,0}(\alpha|\beta) \equiv \sum_{x_\alpha,y_\beta} T_{t,0}(x_\alpha|y_\beta) p_{y|\beta}(0).
\end{equation}
Provided that its inverse exists\footnote{Finding a general answer to the question whether the inverse 
of a dynamical map exists, which allows one to construct a time-local ME, is non-trivial. Nevertheless, many open 
systems can be described by a time-local ME and this assumptions seems to be less strict than one might initially guess. 
See Refs.~\cite{AnderssonCresserHallJMO2007, MaldonadoMundoEtAlPRA2012} for further research on this topic.}, 
it allows to define an effective ME independent of the initial 
mesostate~\cite{HaenggiThomasZPB1977, RivasHuelgaPlenioPRL2010, RivasHuelgaPlenioRPP2014, BreuerEtAlRMP2016}, 
\begin{align}
 \frac{\partial}{\partial t} p_\alpha(t)	&=	\sum_\beta V_{\alpha,\beta}(\lambda_t,t) p_\beta(t),	\label{eq ME meso general}	\\
 V(\lambda_t,t)					            &\equiv	\lim_{\delta t\rightarrow0}\frac{G_{t+\delta t,0}G_{t,0}^{-1}-1}{\delta t},	\label{eq meso generator ME}
\end{align}
but where the matrix $V(\lambda_t,t)$ now carries an additional time-dependence, which does not come from the parameter 
$\lambda_t$. Notice that the construction~(\ref{eq meso generator ME}) shares some similarity with the 
time-convolutionless ME derived from the Nakajima-Zwanzig projection operator formalism, which is another formally exact 
ME independent of the initial mesostate~\cite{FulinskiKramarczyk1968, ShibataTakahashiHashitsumeJSP1977, 
BreuerPetruccioneBook2002, DeVegaAlonsoRMP2017}. %ShimizuJPSJ1970, VstovskyPLA1973, TokuyamaMoriPTP1976, HashitsumaeShibataShinguJSP1977
The generator $V(\lambda_t,t)$ preserves normalization and yields to a set of TMs, which fulfill the Chapman-Kolmogorov 
equation, but it can have temporarily negative rates, i.e., $V_{\alpha,\beta}(\lambda_t,t) < 0$ for $\alpha\neq\beta$ is 
possible. This is a clear indicator that the dynamics are not 1-Markovian~\cite{HallEtAlPRA2014}.  

Finally, we note that there are also other MEs to describe the reduced state of the dynamics, e.g., the standard 
Nakajima-Zwanzig equation which is an integro-differential equation~\cite{BreuerPetruccioneBook2002, DeVegaAlonsoRMP2017}. 
This ME is free from the assumption that the inverse of Eq.~(\ref{eq TM mesolevel general}) exists and therefore 
more general. On the other hand, we will see in Sec.~\ref{sec time dependent MEs} that we will need the notion of 
an IFP of the dynamics, which is hard to define for an integro-differential equation. 

\emph{Lumpability.---} In this final part we introduce the concept of lumpability from Sec.~6.3 in 
Ref.~\cite{KemenySnellBook1976}. It will help us to further understand the conditions which ensure Markovianity at the 
mesolevel and it will be occassionally used in the following. In unison with Ref.~\cite{KemenySnellBook1976} we first 
introduce the concept for discrete, time-homogeneous Markov chains before we consider MEs again. Furthermore, we 
emphasize that in the definition below the notion of Markovianity refers to the usual property~(\ref{eq cond Markovianity}) 
and not only to the one-time probabilities. Another related weaker concept (known as ``weak lumpability'') is treated for 
the interested reader in Appendix~\ref{sec app weak lumpability}.  

\begin{mydef}[Lumpability]\label{def strong lumpability}
 A Markov chain with TM $T_\tau$ is lumpable with respect to a partition $\boldsymbol\chi$ if for every initial 
 distribution $p_{x}(0)$ the lumped process is a Markov chain with transition probabilities independent of $p_{x}(0)$. 
\end{mydef}

It follows from the definition that a lumpable process for a given TM $T_\tau$ and partition $\boldsymbol\chi$, 
is also a lumpable process for all larger times, i.e., for all $T_{n\tau} = (T_\tau)^n$ with $n>1$ and the same 
partition $\boldsymbol\chi$. The following theorem will be useful for us: 

\begin{thm}\label{thm strong lumpability}
 A necessary and sufficient condition for a Markov chain to be lumpable with respect to the partition 
 $\boldsymbol\chi$ is that 
 \begin{equation}\label{eq cond lumpability}
  \C G_{\tau}(\alpha|\beta) \equiv \sum_{x_\alpha} T_{\tau}(x_\alpha|y_\beta) = \sum_{x_\alpha} T_{\tau}(x_\alpha|y'_\beta)
 \end{equation}
 holds for any $y_\beta\neq y'_\beta$. The lumped process then has the TM $\C G_{\tau}$. 
\end{thm}

The details of the proof can be found in Ref.~\cite{KemenySnellBook1976}. However, it is obvious that the so-defined set 
of TMs is independent of the inital state. In addition, one can readily check that they fulfill the 
Chapman-Kolmogorov equation, are normalized and have positive entries. 

The concept of lumpability can be straightforwardly extended to time-dependent MEs by demanding that a lumpable ME 
with respect to the partition $\boldsymbol\chi$ has 
lumpable TMs $T_{t+\delta t,t}$ for any time $t$ and every $\delta t>0$. By expanding Eq.~(\ref{eq cond lumpability}) in 
$\delta t$ and by taking $\delta t\rightarrow0$, we obtain the following corollary (see also Ref.~\cite{NicolisPRE2011}): 

\begin{crllr}\label{thm cor lumpable ME}
 A ME with possibly time-dependent rates is lumpable with respect to the partition $\boldsymbol\chi$ if and only if 
 \begin{equation}\label{eq cond lumpability ME}
  \C V_{\alpha,\beta}(\lambda_t) \equiv \sum_{x_\alpha} W_{x_\alpha,y_\beta}(\lambda_t) = \sum_{x_\alpha} W_{x_\alpha,y'_\beta}(\lambda_t)
 \end{equation}
 for any $y_\beta\neq y'_\beta$ and any $t$. The lumped process is then governed by the rate matrix $\C V(\lambda_t)$. 
\end{crllr}

Notice that the dynamical description of a lumpable ME is unambiguous because the generator 
$R[\lambda_t,p_\alpha(0)]$ from Eq.~(\ref{eq effective rate matrix CG}) and $V(\lambda_t,t)$ from 
Eq.~(\ref{eq meso generator ME}) both coincide with $\C V(\lambda_t)$ from the above corollary. For 
$R[\lambda_t,p_\alpha(0)]$ this follows from directly applying Eq.~(\ref{eq cond lumpability ME}) to 
Eq.~(\ref{eq effective rate matrix CG}). For $V(\lambda_t,t)$ this follows from the fact that the 
propagator in Eq.~(\ref{eq def intermediate TM}) coincides for a Markovian process with the transition 
probabilities obtained from Eq.~(\ref{eq cond Markovianity}), which for a lumpable process are identical to the 
TMs introduced in Theorem~\ref{thm strong lumpability}. All generators are then identical and have the same 
well-defined rate matrix. 

In the following we will stop repeating that any concept at the coarse-grained level is always introduced 
``with respect to the partition $\boldsymbol\chi$''. Furthermore, to facilitate the readability, Table~\ref{table} 
summarizes the most important notation used in this section and in the remainder. 

\begin{table}[h]
 \centering
  \begin{tabular}{l|l}
   symbol				&	meaning			\\
   \hline 
   $\C X$				&	full state space	\\
   $\boldsymbol\chi$			&	state space partition	\\
   $x$					&	arbitary microstate	\\
   $\alpha$				&	mesostate		\\
   $x_\alpha$				&	microstate belonging to mesostate $\alpha$	\\
   $\pi_x(\lambda_t)$			&	microlevel IFP	\\
   $\pi_\alpha(\lambda_t)$		&	$= \sum_{x_\alpha} \pi_{x_\alpha}(\lambda_t)$ (no IFP in general!)	\\
   $\C A(0)$				&	set of admissible initial states	\\
   $\C A_\pi(\lambda_t)$		&	$\rightarrow$ Eq.~(\ref{eq stationary preparation class}), in general dependent on $\lambda_t$	\\
   $\C A(t)$				&	$\C A(0)$ time-evolved	\\
   $p_x(t)$ [$\rho(x;t)$]		&	microstate probability discrete [continuous]	\\
   $p_\alpha(t)$ [$\rho(\alpha;t)$]	&	mesostate probability discrete [continuous]	\\
   $W(\lambda_t)$			&	rate matrix for microdynamics	\\
   $D[p_x\|q_x]$			&	relative entropy	\\
  \end{tabular}
  \caption{\label{table} List of symbols frequently used in the text. }
\end{table}

%%%%%%%%%%%%%%%%%%%%%%%%%%%%%%%%%%%%%%%%%%%%%%%%%%%%%%%%%%%%%%%%%%%%%%%%%%%%%%%%%%%%%%%%%%%%%%%%%%%%%%%%%%%%%%%%%%%%%%%%
\subsection{Entropy production rates, non-Markovianity and instantaneous fixed points}
\label{sec time dependent MEs}

After having discussed how to describe the dynamics at the mesolevel, we now turn to its thermodynamics. This is still 
done in an abstract way without recourse to an underlying physical model. 
An important concept in our theory is the notion of an IFP, which we define as follows: 

\begin{mydef}[Instantaneous fixed point]\label{def IFP}
 Let $V(\lambda_t,t)$ be the generator of the time-local ME~(\ref{eq ME meso general}). We say that $\tilde{\bs\pi}(t)$ 
 is an IFP of the dynamics if $V(\lambda_t,t)\tilde{\bs\pi}(t) = 0$. 
\end{mydef}

We notice that $\tilde{\bs\pi}(t)$ does not need to be a well-defined probability distribution because $V(\lambda_t,t)$ 
can have negative rates. We also point out that the IFP at time $t$ might not be reachable from any state in the class 
of initially admissible states and it is therefore a purely abstract concept. Hence, while 
$V(\lambda_t,t)\tilde{\bs\pi}(t) = 0$ it need not be true that $R[\lambda_t,p_\alpha(0)]\tilde{\bs\pi}(t) = 0$ for any 
$p_{x_\alpha}(0)\in\C A(0)$. The IFP cannot be computed with the help of the effective rate matrix in 
Eq.~(\ref{eq effective rate matrix CG}). The IFP is only well-defined for a time-local ME with a generator independent 
of the initial mesostate. In Appendix~\ref{sec app IFP} we will show that it also does not matter how we have 
derived the ME as long as it is time-local, formally exact and independent of the initial mesostate. 

In the first part of this section, we introduce the concept of EP rate in a formal way and establish a general theorem. 
In the second part of this section, we will answer the question when does the IFP $\tilde{\bs\pi}(t)$ coincide with the 
marginalized IFP of the microdynamics, 
\begin{equation}\label{eq marginalized IFP}
 \pi_\alpha(\lambda_t) = \sum_{x_\alpha} \pi_{x_\alpha}(\lambda_t).
\end{equation}

\emph{EP rate.---} We define the EP rate for the coarse-grained process by 
\begin{equation}\label{eq ent prod abstract}
 \begin{split}
  \dot\Sigma(t)	&\equiv -\left.\frac{\partial}{\partial t}\right|_{\lambda_t} D[p_\alpha(t)\|\pi_\alpha(\lambda_t)]	\\
		&=	-\sum_\alpha \frac{\partial p_\alpha(t)}{\partial t}[\ln p_\alpha(t) - \ln\pi_\alpha(\lambda_t)],
 \end{split}
\end{equation}
where $\pi_\alpha(\lambda_t)$ was defined in Eq.~(\ref{eq marginalized IFP}).\footnote{We remark that it turns out to be 
important to use in our definition~(\ref{eq ent prod abstract}) the coarse-grained steady state $\pi_\alpha(\lambda_t)$ 
and not the actual IFP $\tilde\pi_\alpha(t)$ of the generator $V(\lambda_t,t)$. In the latter case, the so-defined EP rate 
has only a clear thermodynamic meaning in the Markovian limit, where it was previously identified with the non-adiabatic 
part of the EP rate~\cite{EspositoVanDenBroeckPRE2010, VanDenBroeckEspositoPRE2010}. }
Notice that $\dot\Sigma(t)$ can be defined 
for any stochastic process and \emph{a priori} it is not related to the physical EP rate known from nonequilibrium 
thermodynamics. However, for the systems considered in Secs.~\ref{sec coarse-grained dissipative dynamics} 
and~\ref{sec classical system bath theory} this will turn out to be the case. Having emphasized this point, we decided 
for simplicity to refrain from introducing a new terminology for $\dot\Sigma(t)$ in this section. Furthermore, we remark 
that the definition of $\dot\Sigma(t)$ is experimentally meaningful: it only requires to measure the mesostate 
$p_\alpha(t)$ and the knowledge of $\pi_\alpha(\lambda_t)$. The latter can be obtained by measuring the steady 
state of the system after holding $\lambda_t$ fixed for a long time or by arguments of equilibrium statistical 
mechanics (see Secs.~\ref{sec coarse-grained dissipative dynamics} and~\ref{sec classical system bath theory}). 
Also theoretically, Eq.~(\ref{eq ent prod abstract}) can be evaluated with any method that 
gives the exact evolution of the mesostates. 

The following theorem shows how to connect negative EP rates to non-Markovianity. 
Application of this theorem to various physical situations will be the purpose of the next sections. 

\begin{thm}\label{thm ent prod}
 If $\pi_\alpha(\lambda_t)$ is an IFP of the mesodynamics and if $I$ 
 denotes the time interval in which the mesodynamics are 1-Markovian, then $\dot\Sigma(t) \ge 0$ for all $t\in I$. 
\end{thm}

To prove this theorem, it is useful to recall the well-known lemma, which we have stated already in 
Eq.~(\ref{eq contractivity}): 

\begin{lemma}\label{lemma Markov contractivity}
 For a 1-Markovian process the relative entropy between any two probability distributions is continuously decreasing in 
 time, i.e., for all $t$ and any pair of intial distributions $p_\alpha(0)$ and $q_\alpha(0)$ 
 Eq.~(\ref{eq contractivity}) holds. 
\end{lemma}

This lemma follows from the fact that, firstly, for every stochastic matrix $M$ and any pair of distributions 
$p_\alpha$ and $q_\alpha$ one has that 
\begin{equation}\label{eq contractivity rel ent}
 D\left[\sum_\beta M_{\alpha,\beta}p_\beta\left\|\sum_\beta M_{\alpha,\beta}q_\beta\right]\right. \le D[p_\alpha\|q_\alpha],
\end{equation}
and secondly, for a 1-Markovian process the TM at any time $t$ and for every time step $\delta t$ is stochastic. 
We can now prove Theorem~\ref{thm ent prod}: 

\begin{proof}
 By definition of the EP rate we have 
 \begin{align}
  & \dot\Sigma(t) =     \label{eq help 3}	\\
  & -\lim_{\delta t\rightarrow0}\frac{D[G_{t+\delta t,t}\bb p_\text{cg}(t)\|\bs\pi_\text{cg}(\lambda_t)]-D[\bb p_\text{cg}(t)\|\bs\pi_\text{cg}(\lambda_t)]}{\delta t},	\nonumber
 \end{align}
 where $G_{t+\delta t,t}$ is the propagator obtained from the ME~(\ref{eq ME meso general}) [cf.~also 
 Eq.~(\ref{eq def intermediate TM})], $\bb p_\text{cg}(t)$ denotes the vector of the coarse-grained state $p_\alpha(t)$ 
 and likewise for $\bs\pi_\text{cg}(\lambda_t)$. Next, 
 we use the assumption that $\bs\pi_\text{cg}(\lambda_t)$ is an IFP 
 of the ME~(\ref{eq ME meso general}), i.e., we have 
 \begin{equation}
  G_{t+\delta t,t}\bs\pi_\text{cg}(\lambda_t) \approx \bs\pi_\text{cg}(\lambda_t)
 \end{equation}
 and any possible discrepancy vanishs in the limit $\delta t\rightarrow0$. Thus, 
 we can rewrite Eq.~(\ref{eq help 3}) 
 \begin{equation}
  \begin{split}
   \dot\Sigma(t) = -\lim_{\delta t\rightarrow0}\frac{1}{\delta t}\big\{ 	&	D[G_{t+\delta t,t}\bb p_\text{cg}(t)\|G_{t+\delta t,t}\bs\pi_\text{cg}(\lambda_t)]	\\
							&	-D[\bb p_\text{cg}(t)\|\bs\pi_\text{cg}(\lambda_t)]\big\}.
  \end{split}
 \end{equation}
 Now, if the dynamics is 1-Markovian (Definition~\ref{def 1 Markovian}), then $G_{t+\delta t,t}$ 
 is a stochastic matrix and from Eq.~(\ref{eq contractivity rel ent}) it follows that $\dot\Sigma(t) \ge 0$. 
\end{proof}

Whereas the proof of Theorem~\ref{thm ent prod} is straightforward, two things make it a non-trivial statement. 
First, we will show that the EP rate defined in Eq.~(\ref{eq ent prod abstract}) deserves its name because it can be 
linked to \emph{physical} quantities with a \emph{precise thermodynamic interpretation}. This will be done in 
Secs.~\ref{sec coarse-grained dissipative dynamics} and~\ref{sec classical system bath theory}. Second, the essential 
assumption that $\pi_\alpha(\lambda_t)$ is an IFP of the mesodynamics is non-trivial: it is \emph{not} a consequence 
of a 1-Markovian time-evolution and it can also happen for \emph{non}-Markovian dynamics. The details of this crucial 
assumption will be worked out in the remainder of this section, but already at this point we emphasize that 
1-Markovianity alone is \emph{not} sufficient to guarantee that $\dot\Sigma(t) \ge 0$. 
The Venn diagramm in Fig.~\ref{fig Venn} should help to understand the implications of Theorem~\ref{thm ent prod} better. 

\begin{figure}%[h]
 \centering\includegraphics[width=0.40\textwidth,clip=true]{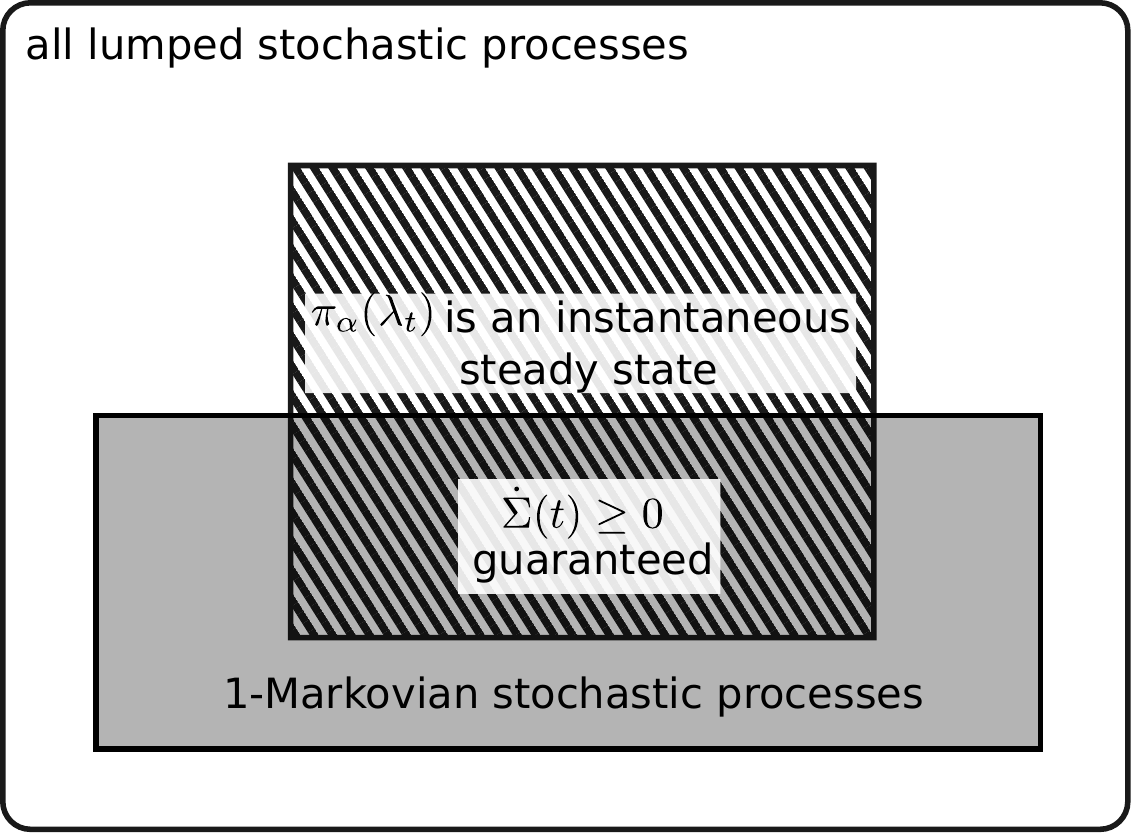}
 \label{fig Venn} 
 \caption{A Venn-diagramm to understand the implications of Theorem~\ref{thm ent prod}. The largest outer box contains 
 all possible lumped stochastic process. One subset of them is 1-Markovian (shaded in grey). For another subset 
 the maginalized microlevel steady state $\pi_\alpha(\lambda_t)$ is an IFP of the dynamics (striped area). Where both 
 sets overlap, $\dot\Sigma(t)\ge0$ is guaranteed, i.e., whenever we observe $\dot\Sigma(t) < 0$ we cannot be 
 simultaneously in the striped and in the shaded grey area. Note that the Venn-diagramm shows the situation for a 
 fixed driving protocol $\lambda_t$ and time interval $I$. Depending on $\lambda_t$ and $I$ the shaded grey and the 
 striped area can parametrically change in time. }
\end{figure}

\emph{IFP of the coarse-grained process.---} To answer the question when is $\tilde\pi_\alpha(\lambda_t) = \pi_\alpha(t)$, 
we start with the simple case and assume that the coarse-grained dynamics are lumpable. Hence, according to 
Corollary~\ref{thm cor lumpable ME} there is a unique and well-defined rate matrix. We then get: 

\begin{thm}\label{thm steady state strong}
 If the stochastic process is lumpable for some time interval $I$, then the IFP of the 
 mesostates is given by the marginal IFP of $W(\lambda_t)$ for all $t\in I$. 
\end{thm}

\begin{proof}
 We want to show that $\C V(\lambda_t)\bs\pi(\lambda_t) = 0$. By using Corollary~\ref{thm cor lumpable ME} 
 in the first and third equality, we obtain 
 \begin{equation}
  \begin{split}
   \sum_\beta \C V_{\alpha,\beta}(\lambda_t)\pi_\beta(\lambda_t)    &=  \sum_\beta\sum_{x_\alpha} W_{x_\alpha,y_\beta}(\lambda_t)\pi_\beta(\lambda_t)    \\
                                                                    &=  \sum_\beta\sum_{x_\alpha,y'_\beta} W_{x_\alpha,y_\beta}(\lambda_t)\pi_{y'_\beta}(\lambda_t)    \\
                                                                    &=  \sum_\beta\sum_{x_\alpha,y'_\beta} W_{x_\alpha,y'_\beta}(\lambda_t)\pi_{y'_\beta}(\lambda_t),
  \end{split}
 \end{equation}
 which is zero since $\pi_{y'_\beta}(\lambda_t)$ is the IFP at the microlevel. 
\end{proof}

Therefore, together with Theorem~\ref{thm ent prod} we can infer that $\dot\Sigma(t) < 0$ unambiguously shows that the 
dynamics are not lumpable. However, lumpability required the coarse-grained process to fulfill the 
Markov property~(\ref{eq cond Markovianity}) for any initial condition, which is a rather strong property. 
We are therefore interested whether a negative EP rate reveals also insights about the weaker property of 
1-Markovianity. For instance, for undriven processes we intuitively expect that, provided that we start at steady state, 
we always remain at steady state independently of the time-dependence of the generator~(\ref{eq meso generator ME}) or 
even the question whether the inverse of Eq.~(\ref{eq TM mesolevel general}) exists. Then, negative values of the EP rate 
will always indicate non-Markovian dynamics for undriven system. Indeed, the following theorem holds: 

\begin{thm}\label{thm steady state weak}
 Consider an undriven stochastic process described by the ME~(\ref{eq ME meso general}), i.e., we assume $G_{t,0}^{-1}$ 
 to exist for all admissible initial states $\C A(0)$ and all times $t$. If the conditional microstates are initially 
 equilibrated, $\C A(0) \subset\C A_\pi$ [Eq.~(\ref{eq stationary preparation class})], then 
 $\pi_\alpha$ is an IFP of the stochastic process at the mesolevel. 
\end{thm}

\begin{proof}
 If $\C A(0) \subset\C A_\pi$, we can conclude that $\sum_\beta G_{t,0}(\alpha|\beta)\pi_\beta = \pi_\alpha$, i.e., 
 if we start with the coarse-grained steady state we also remain in it for all times $t$. Since $G_{t,0}$ was assumed 
 to be invertible, 
 \begin{equation}\label{eq invertible steady state}
  \sum_\beta (G_{t,0}^{-1})_{\alpha,\beta}\pi_\beta = \pi_\alpha.
 \end{equation}
 Hence, by definition~(\ref{eq meso generator ME}) we obtain the chain of equalities 
 \begin{equation}
  \begin{split}\label{eq help thm steady state weak}
   & \sum_{\alpha,\beta}\left(\lim_{\delta t\rightarrow0}\frac{G_{t+\delta t,0}G_{t,0}^{-1} - 1}{\delta t}\right)_{\alpha,\beta} \pi_\beta	\\
   & = \lim_{\delta t\rightarrow0}\frac{1}{\delta t}\left[\sum_{\beta,\gamma} G_{t+\delta t,0}(\alpha|\gamma)(G_{t,0}^{-1})_{\gamma,\beta}\pi_\beta - \pi_\alpha\right]	\\
   & = \lim_{\delta t\rightarrow0}\frac{1}{\delta t}\left[\sum_{\gamma} G_{t+\delta t,0}(\alpha|\gamma)\pi_\gamma - \pi_\alpha\right]	\\
   & = \lim_{\delta t\rightarrow0}\frac{1}{\delta t}\left[\pi_\alpha - \pi_\alpha\right] = 0.
  \end{split}
 \end{equation}
\end{proof}

We recognize a big difference in the characterization of the IFPs for driven and undriven processes. 
Without driving, the right set of initial states suffices already to show that the microlevel steady state induces the 
steady state at the mesolevel, even if the dynamics is non-Markovian. Thus, for this kind of dynamics 
$\dot\Sigma(t) < 0$ unambiguously signifies non-Markovianity. For driven systems instead, we needed the much 
stronger requirement of lumpability, i.e., Markovianity of the lumped process with TMs independent 
of the initial microstate. However, at least formally it is possible to establish the following additional theorem: 

\begin{thm}\label{thm steady state general}
 Consider a driven stochastic process described by the ME~(\ref{eq ME meso general}), i.e., we assume $G_{t,0}^{-1}$ to 
 exist for all initial states and all times $t$. We denote by $I$ the time-interval in which either 
 \begin{enumerate}
  \item all conditional microstates in the set of time-evolved states are at steady state, 
  $\C A(t)\subset\C A_\pi(\lambda_t)$, or 
  \item the IFP of the microdynamics is an admissible time-evolved state, $\pi_x(\lambda_t)\in\C A(t)$. 
 \end{enumerate}
 Then, $\pi_{\alpha}(\lambda_t)$ is an IFP of the lumped process for all $t\in I$. 
\end{thm}

\begin{proof}
 First of all, notice that the ME~(\ref{eq ME meso general}) generates the exact time evolution, i.e., for any 
 $p_y(t) = p_\beta(t)p_{y|\beta}(t)\in\C A(t)$ we have 
 \begin{equation}
  \begin{split}\label{eq help 4}
   & \sum_\beta V_{\alpha,\beta}(\lambda_t,t) p_\beta(t)	\\
   & = \sum_{x_\alpha}\sum_{\beta,y_\beta} W_{x_\alpha,y_\beta}(\lambda_t) p_{y|\beta}(t)p_\beta(t).
  \end{split}
 \end{equation}
 For the first condition, if $\pi_x(\lambda_t)\in\C A(t) \subset \C A_\pi(\lambda_t)$, then one immediately verifies that 
 $V(\lambda_t,t)\bs\pi(\lambda_t) = 0$. But one may have that $\C A(t) \subset \C A_\pi(\lambda_t)$, 
 but $\pi_x(\lambda_t)\notin\C A(t)$. This means that there is no admissible initial state, which gets mapped to the IFP 
at time $t$, i.e., $T_{t,0}^{-1}\bs{\pi}(\lambda_t) \notin\C A(0)$. 
 However, by the invertibility of the dynamics there is always a set of states $p_x^{(i)}(t)\in\C A(t)$, 
 which spans the entire mesostate space. Thus, we can always find a linear combination 
 $\pi_x(\lambda_t) = \sum_i \mu_i p_x^{(i)}(t)$ with $\mu_i\in\mathbb{R}$. 
 Then, $V(\lambda_t,t)\bs\pi(\lambda_t) = 0$ follows from the linearity of the dynamics by applying 
Eq.~(\ref{eq help 4}) to each term of the linear combination. 
 
 For the second condition let us assume the opposite, i.e., $V(\lambda_t,t)\bs\pi(\lambda_t) \neq 0$. This implies 
 $\sum_\beta G_{t+\delta t,t}(\alpha|\beta) \pi_\beta(\lambda_t) \neq \pi_\alpha(\lambda_t)$ for a sufficiently 
 small $\delta t$. But as the reduced dynamics are exact, this can only be the case if there is a state 
 $q_y(t) = \pi_\beta(\lambda_t)q_{y|\beta}(t)\in\C A(t)$ with $q_{y|\beta}(t)\neq \pi_{y|\beta}(\lambda_t)$. 
 On the other hand, the theorem assumes that $\pi_x(\lambda_t)\in\C A(t)$ too. Hence, 
 there must be two states $q_y(t)\in\C A(t)$ and $\pi_x(\lambda_t)\in\C A(t)$, which give the 
 same marginal mesostate $\pi_\alpha(\lambda_t)$. Since the ME dynamics in the full space are clearly invertible 
 and since the initial conditional microstate is fixed, this means that there must be two different initial 
 mesostates, which get mapped to the same mesostate at time $t$. Hence, $G_{t,0}$ cannot be invertible, which 
 conflicts with our initial assumption. 
\end{proof}

Theorem~\ref{thm steady state general} plays an important role in the limit of TSS (see 
Sec.~\ref{sec time scale separation}) where the first condition is automatically fulfilled. The second condition 
will be in general complicated to check if the microdynamics are complex. 

It is worthwhile to ask whether milder conditions suffice to ensure that $\pi_\alpha(\lambda_t)$ is an IFP of the 
mesodynamics. In Appendix~\ref{sec app weak lumpability} we show that they can indeed be found if the dynamics fulfills 
the special property of weak lumpability. In general, however, we believe that it will be hard to find milder conditons: 
in Sec.~\ref{sec steady states} we give an example for an ergodic and undriven Markov chain, whose mesodynamics are 
1-Markovian, but $\pi_\alpha$ is not an IFP unless $\C A(0)\subset\C A_\pi$. As any driven process takes the conditional 
microstates out of equilibrium, i.e., $\C A(t) \nsubseteq \C A_\pi(\lambda_t)$ in general, finding useful milder conditions 
to guarantee that $\pi_\alpha(\lambda_t)$ is an IFP seems unrealistic. 

Before we proceed with the physical picture, we want to comment on a mathematical subtlety, which becomes relevant for 
the application considered in Sec.~\ref{sec classical system bath theory}. In there, we will apply our findings from 
above to the case of Hamiltonian dynamics described on the continuous phase space of a collection of classical 
particles. This does not fit into the conventional picture of a finite and discrete state space $\C X$ with $N<\infty$ 
microstates. However, under the assumption that it is possible to approximate the actual Hamiltonian dynamics by using a 
high-dimensional grid of very small phase space cells, we can imagine that we can approximate the true dynamics 
arbitrarily well with a finite, discretized phase space. Nevertheless, in order not to rely on this way of reasoning, we 
briefly re-derive the above theorems for the Hamiltonian setting in Appendix~\ref{sec app Hamiltonian dynamics}.

%%%%%%%%%%%%%%%%%%%%%%%%%%%%%%%%%%%%%%%%%%%%%%%%%%%%%%%%%%%%%%%%%%%%%%%%%%%%%%%%%%%%%%%%%%%%%%%%%%%%%%%%%%%%%%%%%%%%%%%%
\section{Coarse-grained dissipative dynamics}
\label{sec coarse-grained dissipative dynamics}

%%%%%%%%%%%%%%%%%%%%%%%%%%%%%%%%%%%%%%%%%%%%%%%%%%%%%%%%%%%%%%%%%%%%%%%%%%%%%%%%%%%%%%%%%%%%%%%%%%%%%%%%%%%%%%%%%%%%%%%%
\subsection{Thermodynamics at the microlevel}
\label{sec thermodynamics microlevel}

We now start to investigate the first application of the general framework from Sec.~\ref{sec mathematical results}. 
In this section we consider the ME~(\ref{eq ME general}), which describes a large class of dissipative classical and 
quantum systems, with applications ranging from molecular motors to thermoelectric devices. In addition, we impose the 
condition of local detailed balance, 
\begin{equation}\label{eq local detailed balance}
 \ln\frac{W_{x,y}(\lambda_t)}{W_{y,x}(\lambda_t)} = -\beta[E_x(\lambda_t) - E_y(\lambda_t)],
\end{equation}
where $E_x(\lambda_t)$ denotes the energy of state $x$ and $\beta$ the inverse temperature of the bath. 
Eq.~(\ref{eq local detailed balance}) ensures that the IFP at the microlevel is given by the Gibbs state 
$\pi_x(\lambda_t) = e^{-\beta E_x(\lambda_t)}/Z(\lambda_t)$ with $Z(\lambda_t) = \sum_x e^{-\beta E_x(\lambda_t)}$ and 
it allows us to link energetic changes in the system with entropic changes in the bath. A thermodynamically 
consistent description of the microdynamics follows from the definitions 
\begin{align}
 U_\text{mic}(t)		&\equiv	\sum_x E_x(\lambda_t) p_x(t) ~ (\text{internal energy}),	\label{eq U mic}	\\
 \dot W_\text{mic}(t)		&\equiv \sum_x [d_t E_{x}(\lambda_t)]p_x(t) ~ (\text{work rate}),	\label{eq work}	\\
 \dot Q_\text{mic}(t)		&\equiv	\sum_{x} E_{x}(\lambda_t) \partial_tp_{x}(t) ~ (\text{heat rate}),	\label{eq heat}	\\
 S_\text{mic}(t)		&\equiv -\sum_x p_x(t)\ln p_{x}(t) ~ (\text{Shannon entropy}),	\label{eq S mic}	\\
 F_\text{mic}(t)		&\equiv	U_\text{mic}(t) - S_\text{mic}(t)/\beta ~ (\text{free energy}), \label{eq F mic}	\\
 \dot\Sigma_\text{mic}(t)	&\equiv	-\left.\frac{\partial}{\partial t}\right|_{\lambda_t} D[p_x(t)\|\pi_x(\lambda_t)] \ge 0 ~ (\text{EP rate}).
\end{align}
Here, we used the subscript ``mic'' to emphasize that the above definitions refer to the thermodynamic description 
of the microdynamics, which has to be distinguished from the thermodynamic description at the mesolevel introduced 
below. Using the ME~(\ref{eq ME general}) and local detailed balance~(\ref{eq local detailed balance}) together with the 
definitions provided above, one can verify the first and second law of thermodynamics in the conventional form: 
$d_t U_\text{mic}(t) = \dot W_\text{mic}(t) + \dot Q_\text{mic}(t)$ and 
$\dot\Sigma_\text{mic}(t) = \beta[\dot W_\text{mic}(t) - d_t F_\text{mic}(t)] \ge 0$. 

Since the IFP at the microlevel is the equilibrium Gibbs state, we can parametrize the conditional equilibrium state of 
the microstates belonging to a mesostate $\alpha$ as 
\begin{equation}\label{eq cond eq state CG dynamics}
 \pi_{x|\alpha}(\lambda_t) = e^{-\beta[E_{x_\alpha}(\lambda_t) - F_\alpha(\lambda_t)]},
\end{equation}
where $F_\alpha(\lambda_t) \equiv -\beta^{-1}\ln\sum_{x_\alpha} e^{-\beta E_{x_\alpha}(\lambda_t)}$ plays the role of 
an effective free energy. The reduced equilibrium distribution of a mesostate can then be written as 
\begin{equation}\label{eq state CG dynamics}
 \pi_\alpha(\lambda_t) = \frac{e^{-\beta F_\alpha(\lambda_t)}}{Z(\lambda_t)}.
\end{equation}

In the following we want to find meaningful definitions, which allow us to formulate the laws of thermodynamics at 
a coarse-grained level and which we can connect to the general theory of Sec.~\ref{sec mathematical results}. 
Since the dynamics at the mesolevel will typically be non-Markovian and not fulfill local detailed balance, finding a 
consistent thermodynamic framework becomes non-trivial. We will restrict our investigations here to any initial 
prepartion class which fulfills $\C A(0)\subset\C A_\pi(\lambda_0)$ with $\C A_\pi(\lambda_0)$ defined in 
Eq.~(\ref{eq stationary preparation class}). If the dynamics is driven, we will need one additional assumption 
[see Eq.~(\ref{eq cond bipartite driving})], otherwise our results are general.

%%%%%%%%%%%%%%%%%%%%%%%%%%%%%%%%%%%%%%%%%%%%%%%%%%%%%%%%%%%%%%%%%%%%%%%%%%%%%%%%%%%%%%%%%%%%%%%%%%%%%%%%%%%%%%%%%%%%%%%%
\subsection{Thermodynamics at the mesolevel}

With the framework from Sec.~\ref{sec mathematical results} we are now going to study the thermodynamics at the 
mesolevel. This is possible in full generality if the dynamics are undriven. In case of driving, 
$\dot\lambda_t\neq0$, we need to assume that  we can split the time-dependent energy function as 
\begin{equation}\label{eq cond bipartite driving}
 E_{x_\alpha}(\lambda_t) = E_{\alpha}(\lambda_t) + \tilde E_{x_\alpha}.
\end{equation}
Thus, solely the mesostate energies are affected by the driving. This condition naturally arises if we think about the 
complete system as being composed of two interacting systems, $\C X = \C Y\otimes \C Z$, and we trace out the degrees of 
freedom $\C Y$ to obtain a reduced description in $\C Z$. In this case we can split the energy for any value of 
$\lambda_t$ as $E_{yz} = E_y + E_z + V_{yz}$ where $V_{yz}$ describes an interaction energy and $E_y$ ($E_z$) are the 
bare energies associated with the isolated system $\C Y$ ($\C Z$). Condition~(\ref{eq cond bipartite driving}) is then 
naturally fulfilled if we identify $E_z = E_\alpha$ and only $E_z = E_z(\lambda_t)$ is time-dependent (compare also with 
Sec.~\ref{sec classical system bath theory}). Importantly, this condition allows us to identify 
\begin{equation}
 \begin{split}\label{eq work bipartite}
  \dot W_\text{mic}(t)	&=	\sum_x \frac{\partial E_x(\lambda_t)}{\partial t}p_x(t)	\\
			&=	\sum_\alpha \frac{\partial E_\alpha(\lambda_t)}{\partial t} p_\alpha(t) \equiv \dot W(t).
 \end{split}
\end{equation}
Therefore, the exact rate of work can be computed from the knowledge about the mesostate alone. Furthermore, 
Eq.~(\ref{eq cond bipartite driving}) implies that the conditional equilibrium state of the 
bath~(\ref{eq cond eq state CG dynamics}) does not depend on $\lambda_t$ and hence, we can write 
$\C A_\pi(\lambda_0) = \C A_\pi$. 

The thermodynamic analysis starts from our central definition~(\ref{eq ent prod abstract}) 
\begin{equation}
 \dot\Sigma(t) = -\left.\frac{\partial}{\partial t}\right|_{\lambda_t} D[p_\alpha(t)\|\pi_\alpha(\lambda_t)]
\end{equation}
with $\pi_\alpha(\lambda_t)$ given in Eq.~(\ref{eq state CG dynamics}). Using Eq.~(\ref{eq work bipartite}) 
and noting that $d_t F_\alpha(\lambda_t) = d_t E_\alpha(\lambda_t)$, it is not hard to confirm that 
\begin{equation}
 \dot\Sigma(t) = \beta\dot W(t) - \beta\frac{d}{dt}\sum_\alpha p_\alpha(t)\left[F_\alpha(\lambda_t) + \frac{1}{\beta}\ln p_\alpha(t)\right].
\end{equation}
This motivates the definition of the nonequilibrium free energy 
\begin{equation}
 F(t) \equiv \sum_\alpha p_\alpha(t)\left[F_\alpha(\lambda_t) + \frac{1}{\beta}\ln p_\alpha(t)\right],	\label{eq free energy mesolevel}
\end{equation}
such that the EP rate is given by the familiar form of phenomenological non-equilibrium thermodynamics: 
$\dot\Sigma(t) = \beta[\dot W(t) - d_t F(t)]$. The EP over a finite time interval becomes 
\begin{equation}\label{eq ent prod CG dynamics}
 \Sigma(t) = \beta[W(t) - \Delta F(t)]
\end{equation}
and for a proper second law it remains to show that this quantity is positive. This follows from: 

\begin{thm}
 For any $p_x(0) \in \C A_\pi$ and any driving protocol we have 
 \begin{equation}
  \Sigma(t) \ge \Sigma_\text{\normalfont mic}(t) \ge 0.
 \end{equation}
\end{thm}

\begin{proof}
 The proof was already given in Ref.~\cite{StrasbergEspositoPRE2017}. In short, one rewrites 
 \begin{equation}
  \Sigma(t) - \Sigma_\text{mic}(t) = \beta[\Delta F_\text{mic}(t) - \Delta F(t)]
 \end{equation}
 and shows that for $p_x(0) \in \C A_\pi$ it follows that 
 \begin{equation}
  \begin{split}
   & \beta[\Delta F_\text{mic}(t) - \Delta F(t)]	\\
   & = D[p_x(t)\|\pi_x(\lambda_t)] - D[p_\alpha(t)\|\pi_\alpha(\lambda_t)] \\
   & = \sum_\alpha p_\alpha(t) D[p_{x|\alpha}(t)\|\pi_{x|\alpha}] \ge 0.
  \end{split}
 \end{equation}
 Since $\Sigma_\text{mic}(t)\ge 0$, this implies $\Sigma(t)\ge 0$.  
\end{proof}

Using the theorems of Sec.~\ref{sec time dependent MEs}, we can now connect the appearance of negative 
EP rates to the following properties of the underlying dynamics: 

\begin{thm}\label{thm ent prod bipartite case}
 Let $p_x(0) \in \C A_\pi$ and 
 let $I$ denote the time interval in which the mesodynamics are 1-Markovian and the dynamics is  
 \begin{enumerate}
  \item undriven, or 
  \item driven and lumpable, or 
  \item driven and such that $\C A(t)\subset\C A_\pi$ or $\pi_x(\lambda_t)\in\C A(t)$. 
 \end{enumerate}
 Then, $\dot\Sigma(t) \ge 0$ for all $t\in I$ and all admissible initial states. 
\end{thm}

Hence, as a corollary, if we observe $\dot\Sigma(t) < 0$ for the undriven case, we know that the dynamics 
is non-Markovian [or that the initial state $p_x(0)\notin\C A_\pi$]. For driven dynamics, noticing a negative 
EP rate, is not sufficient to conclude that the dynamics is non-Markovian, but they are clearly not lumpable. 
In the next section we will show that $\dot\Sigma(t) < 0$ also suffices to conlude that TSS does not apply. 

Furthermore, while the above procedure provides a unique way to define a non-equilibrium free energy at the 
mesolevel, it does not fix the definition of the internal energy and entropy at the mesolevel because the 
prescription $F = U-S/\beta$ entails a certain level of arbitrariness. Via the first law $\Delta U = Q + W$ this would 
also imply a certain arbitrariness for the definition of heat~\cite{TalknerHaenggiPRE2016}. However, a reasonable 
definition of $U, S$ and $Q$ can be fixed by demanding that they should coincide with $U_\text{mic}, S_\text{mic}$ 
and $Q_\text{mic}$ in the limit where the microstates are conditionally equilibrated, which is fulfilled in the 
limit of TSS considered in Sec.~\ref{sec time scale separation}. Then, one is naturally lead to the definitions 
\begin{align}
 U(t)	&\equiv	\sum_\alpha \C U_\alpha(\lambda_t) p_\alpha(t), ~~~ \C U_\alpha \equiv \sum_{x_\alpha} E_{x_\alpha}(\lambda_t)\pi_{x|\alpha},	\label{eq int energy mesolevel}	\\
 S(t)	&\equiv	\sum_\alpha \left\{\beta [\C U_\alpha(\lambda_t) - F_\alpha(\lambda_t)] -\ln p_\alpha(t)\right\} p_\alpha(t).	\label{eq entropy mesolevel}
\end{align}
Heat is then defined as $\dot Q(t) = d_t U(t) - \dot W(t)$ and the EP rate can be equivalently expressed as 
$\dot\Sigma(t) = d_tS(t) - \beta\dot Q(t)$. 

We remark that it is not obvious how to relax condition~(\ref{eq cond bipartite driving}) because the 
work~(\ref{eq work bipartite}) can then not be computed from knowledge of the mesostate alone, which was an 
essential ingredient in our derivation.

%%%%%%%%%%%%%%%%%%%%%%%%%%%%%%%%%%%%%%%%%%%%%%%%%%%%%%%%%%%%%%%%%%%%%%%%%%%%%%%%%%%%%%%%%%%%%%%%%%%%%%%%%%%%%%%%%%%%%%%%
\subsection{Time-scale separation and Markovian limits}
\label{sec time scale separation}

Although open systems behave non-Markovian in general, it is important to know in which limits the Markovian 
\emph{approximation} is justified. One such limit is TSS, which is an essential assumption in many branches of 
statistical mechanics in order to ensure that the dynamics at the level of the ``relevant'' degrees of freedom is 
Markovian and hence, easily tractable. It is also essential in order to ensure that we can infer from the coarse-grained 
dynamics the \emph{exact} thermodynamics of the underlying microstate dynamics (under reasonable mild conditions), see 
Refs.~\cite{PuglisiEtAlJSM2010, SeifertEPJE2011, EspositoPRE2012, AltanerVollmerPRL2012, BoCelaniJSM2014, 
StrasbergEspositoPRE2017} for research on this topic. Here, we restrict ourselves to highlight the role of TSS within 
our mathematical framework of Sec.~\ref{sec mathematical results}. Furthermore, at the end of this section we discuss 
another class of systems whose dynamics is Markovian albeit TSS does not apply. 

To study TSS, let us decompose the rate matrix as follows: 
\begin{equation}\label{eq rate matrix TSS}
 W_{x_\alpha,y_\beta}(\lambda_t) = \delta_{\alpha\beta} R_{x_\alpha,y_\alpha}(\lambda_t) + (1-\delta_{\alpha\beta})r_{x_\alpha,y_\beta}(\lambda_t).
\end{equation}
Next, we assume that $R_{x_\alpha,y_\alpha}(\lambda_t) \gg r_{x_\alpha,y_\beta}(\lambda_t)$, 
i.e., there is a strong separation of time-scales between the mesodynamics and the microdynamics belonging 
to a certain mesostate. As a consequence the microstates rapidly equilibrate to the conditional steady state 
$\pi_{x|\alpha}(\lambda_t)$ for any mesostate $\alpha$ provided that the microstates in each mesostate are fully 
connected (tacitly assumed in the following). This means that condition~1 of 
Theorem~\ref{thm steady state general} is always fulfilled. By replacing $p_{y|\beta}(t)$ by 
$\pi_{y|\beta}(\lambda_t)$ in Eq.~(\ref{eq effective rate matrix CG}), it is easy to see that the effective rate 
matrix is independent of the initial state and describes a proper Markov process, 
$R[\lambda_t,p_\alpha(0)] = R(\lambda_t)$. Another consequence of TSS is that the thermodynamics associated with 
the mesodynamics are identical to the thermodynamics of the microdynamics. 

Strictly speaking the limit of TSS requires 
$R_{x_\alpha,y_\alpha}(\lambda_t)/r_{x_\alpha,y_\beta}(\lambda_t) \rightarrow\infty$. In practice, however, there will 
be always a finite time $\delta t$ associated with the relaxation of the microstates and TSS means that we assume 
\begin{equation}
 \frac{1}{r_{x_\alpha,y_\beta}(\lambda_t)} \gg \delta t \gg \frac{1}{R_{x_\alpha,y_\alpha}(\lambda_t)}.
\end{equation}
Then, within a time-step $\delta t$ the conditional microstates are almost equilibrated while terms 
of the order $\C O(\delta t^2 r_{x_\alpha,y_\beta})$ are still negligible. The TM in this situation becomes 
\begin{equation}
 \begin{split}
  & T_{t+\delta t,t}(x_\alpha|y_\beta) \approx	\\
  & \delta_{\alpha\beta} \pi_{x|\alpha}(\lambda_t)\left(1-\delta t\sum_{\gamma\neq\alpha} \sum_{z_\gamma} r_{z_\gamma,x_\alpha}(\lambda_t)\right)	\\
  & + \delta t (1-\delta_{\alpha\beta}) \sum_{z_\beta} \pi_{z|\beta}(\lambda_t) r_{x_\alpha,z_\beta}(\lambda_t).
 \end{split}
\end{equation}
The first term describes the probability for a transition within two microstates of the same mesostate: to lowest 
order this is simply given by the conditional steady state minus a small correction term of $\C O(\delta t)$, which 
takes into account the possibility that one leaves the given mesostate to another mesostate. The second term gives the 
probability to reach a microstate lying in a different mesostate, which is given by the sum of all possible rates which 
connect to this microstate from the given mesostate multiplied by the respective conditional steady state probability. 
One immediately checks normalization of $T_{t+\delta t,t}(x_\alpha|y_\beta)$ and positivity follows by assuming 
that $r_{z_\gamma,x_\alpha}(\lambda_t)\delta t \ll 1$. Furthermore, also the condition~(\ref{eq cond lumpability}) of 
lumpability is fulfilled. Indeed, we can even confirm the stronger property 
\begin{equation}
 T_{t+\delta t,t}(x_\alpha|y_\beta) = T_{t+\delta t,t}(x_\alpha|y'_\beta)
\end{equation}
for all $y'_\beta\neq y_\beta$. Hence, in the idealized limit yielding to an instantaneous equilibration of the 
conditional microstates, the TMs do not even depend on the particular microstate anymore. We conclude: 

\begin{thm}\label{thm TSS}
 If TSS applies, then the process is lumpable and $p_x(t) \in \C A_\pi$ for all $t$. 
 Conversely, if $\dot\Sigma(t) < 0$, then TSS does not apply. 
\end{thm}

It was shown in Ref.~\cite{EspositoPRE2012} that $\dot\Sigma(t) = \dot\Sigma_\text{mic}(t)$ in the limit of TSS. 
If only the slightly weaker condition of lumpabibility is fulfilled, then it is not known whether 
$\dot\Sigma(t) = \dot\Sigma_\text{mic}(t)$ still holds. 

While TSS is an important limit, the mesodynamics can be also Markovian without the assumption of TSS. The following 
theorem demonstrates this explicitly: 

\begin{thm}
 If there is a partition $\boldsymbol\chi$ such that the rate matrix can be written as 
 \begin{equation}\label{eq rate matrix decomposition lumpable}
  W_{x_\alpha,y_\beta}(\lambda_t) = \delta_{\alpha\beta} R_{x_\alpha,y_\alpha}(\lambda_t) + (1-\delta_{\alpha\beta})V_{\alpha,\beta}(\lambda_t),
 \end{equation}
 then the process is lumpable independent of any TSS argument. Moreoever, the IFP of the lumped process is 
 $\pi_\alpha(\lambda_t) = \sum_{x_\alpha} \pi_{x_\alpha}(\lambda_t)$ and hence, $\dot\Sigma(t) \ge 0$ always. 
\end{thm}

\begin{proof}
 We first of all observe that from 
 \begin{equation}
  \begin{split}
   0	&=	\sum_{\alpha,x_\alpha} W_{x_\alpha,y_\beta}(\lambda_t)	\\
	&=	\sum_{x_\beta} R_{x_\beta,y_\beta}(\lambda_t) + \sum_{\alpha\neq\beta}\sum_{x_\alpha} V_{\alpha,\beta}(\lambda_t),
  \end{split}
 \end{equation}
 it follows that 
 $\sum_{x_\alpha} R_{x_\alpha,y_\alpha}(\lambda_t) = -\sum_{\beta\neq\alpha} \#\chi_\beta V_{\beta,\alpha}(\lambda_t)$ 
 for any $\alpha$ (where $\#\chi_\alpha$ denotes the cardinality of the set of microstates belonging to mesostate 
 $\alpha$). By using this property, it becomes straightforward to check that Eq.~(\ref{eq cond lumpability ME}) is 
 fulfilled and hence, the coarse-grained process is Markovian. Due to Theorem~\ref{thm steady state strong} we can 
 also confirm that $\pi_\alpha(\lambda_t)$ is the IFP and from Theorem~\ref{thm ent prod} it follows that 
 $\dot\Sigma(t) \ge 0$. 
\end{proof}

Compared to the decomposition~(\ref{eq rate matrix TSS}) we here did not need to assume any particular scaling of the 
rates, but it was important that the transitions between different mesostates are independent of the microstate. 
In fact, for many mesoscopic systems the details of the microstates might not matter, for instance, the Brownian 
motion of a suspended particle is quite independent from the spin degrees of freedom of its electrons unless strong 
magnetic interactions are present. 
Notice that the ME at the mesolevel resulting from Eq.~(\ref{eq rate matrix decomposition lumpable}) reads 
\begin{equation}
 \begin{split}\label{eq help 1}
  & \frac{\partial}{\partial t} p_\alpha(t) = \\
  &	\sum_{\beta\neq\alpha}\left[\#\chi_\alpha V_{\alpha,\beta}(\lambda_t) p_\beta(t) - \#\chi_\beta V_{\beta,\alpha}(\lambda_t) p_\alpha(t)\right].
 \end{split}
\end{equation}
It shows that the local detailed balance ratio~(\ref{eq local detailed balance}) of the effective rates at the mesolevel 
is shifted by an entropic contribution due to the degeneracy factor $\#\chi_\alpha$; see 
Sec.~\ref{sec strong lumpability without TSS} or Ref.~\cite{HerpichThingnaEspositoPRX2018} for explicit examples.

%%%%%%%%%%%%%%%%%%%%%%%%%%%%%%%%%%%%%%%%%%%%%%%%%%%%%%%%%%%%%%%%%%%%%%%%%%%%%%%%%%%%%%%%%%%%%%%%%%%%%%%%%%%%%%%%%%%%%%%%
\section{Classical system-bath theory}
\label{sec classical system bath theory}

In this section we consider the standard paradigm of classical open system theory: a system in contact with a bath 
described by Hamiltonian dynamics as opposed to the rate ME dynamics from 
Sec.~\ref{sec coarse-grained dissipative dynamics}. The microstates (system and bath) therefore describe an isolated 
system and the goal is to find a consistent thermodynamic framework for the mesostate (the system only). The global 
Hamiltonian reads 
\begin{equation}\label{eq Hamiltonian SB}
 H_\text{tot}(\lambda_t) = H(\lambda_t) + V + H_B,
\end{equation}
where the system, bath and interaction Hamiltonian $H(\lambda_t)$, $H_B$ and $V$ are arbitrary. We denote a phase 
space point of the system by $x_S$ and of the bath by $x_B$. Thus, to be very precise, we should write 
$H(x_S;\lambda_t)$, $H_B(x_B)$ and $V(x_S,x_B)$, but we will drop the dependency on $x_S$ and $x_B$ for notational 
simplicity. Deriving the laws of thermodynamics for an arbitrary Hamiltonian~(\ref{eq Hamiltonian SB}) has attracted 
much interest recently~\cite{JarzynskiJSM2004, GelinThossPRE2009, SeifertPRL2016, 
TalknerHaenggiPRE2016, JarzynskiPRX2017, MillerAndersPRE2017, StrasbergEspositoPRE2017, AurellEnt2017} (note that many 
investigations in the quantum domain also have a direct analogue in the classical 
regime~\cite{EspositoLindenbergVandenBroeckNJP2010, MartinezPazPRl2013, PucciEspositoPelitiJSM2013, StrasbergEtAlNJP2016, 
FreitasPazPRE2017, PerarnauLlobetEtAlPRL2018, HsiangEtAlPRE2018}). It will turn out that our basic definitions are 
identical to the ones suggested by Seifert~\cite{SeifertPRL2016}. We here re-derive them in a different way and in 
addition, we focus on the EP \emph{rate} and its relation to non-Markovian dynamics. 

In order to be able to define the EP rate~(\ref{eq ent prod abstract}), we first of all need to know the exact equilibrium 
state of the system, which is obtained from coarse-graining the global equilibrium state 
$\pi_\text{tot}(\lambda_t) = e^{-\beta H_\text{tot}(\lambda_t)}/\C Z_\text{tot}(\lambda_t)$ with 
$\C Z_\text{tot}(\lambda_t) = \int dx_{SB} e^{-\beta H_\text{tot}(\lambda_t)}$. For this purpose we introduce the 
Hamiltonian of mean force $H^*(\lambda_t)$~\cite{KirkwoodJCP1935}. It is defined through the two relations 
\begin{equation}
 \begin{split}\label{eq HMF}
  \pi_S(\lambda_t)	&\equiv	\frac{e^{-\beta H^*(\lambda_t)}}{\C Z^*(\lambda_t)} = \int dx_B \frac{e^{-\beta H_\text{tot}(\lambda_t)}}{\C Z_\text{tot}(\lambda_t)},	\\
  \C Z^*(\lambda_t)	&\equiv	\frac{\C Z_\text{tot}(\lambda_t)}{\C Z_B},
 \end{split}
\end{equation}
where $\C Z_B = \int dx_B e^{-\beta H_B}$ is the equilibrium partition function of the unperturbed bath. We emphasize 
that the equilibrium state of the system is not a Gibbs state with respect to $H(\lambda_t)$ due to the strong coupling. 
More explicitly, the Hamiltonian of mean force reads
\begin{equation}
 H^*(\lambda_t) = H(\lambda_t) - \frac{1}{\beta}\ln\lr{e^{-\beta V}}_B^\text{eq},
\end{equation}
where $\lr{\dots}_B^\text{eq}$ denotes an average with respect to the unperturbed equilibrium state of the bath 
$e^{-\beta H_B}/\C Z_B$. Note that $H^*(\lambda_t)$ also depends on the inverse temperature $\beta$ of the bath. 

We can now use Eq.~(\ref{eq ent prod abstract}) to define the EP rate, which reads in the notation of this section 
\begin{equation}\label{eq ent prod rate Seifert}
 \dot\Sigma(t) = -\left.\frac{\partial}{\partial t}\right|_{\lambda_t} D[\rho_S(t)\|\pi_S(\lambda_t)],
\end{equation}
where $\rho_S(t) = \rho_S(x_S;t)$ denotes the state of the system at time $t$, which can be arbitrarily far from 
equilibrium. Note that we now use the differential relative entropy 
$D[\rho_S(t)\|\pi_S(\lambda_t)] = \int dx_S \rho_S(x_S;t)\ln\frac{\rho_S(x_S;t)}{\pi_S(x_S;\lambda_t)}$. 
Using Eq.~(\ref{eq HMF}), we can rewrite Eq.~(\ref{eq ent prod rate Seifert}) as 
\begin{equation}
 \dot\Sigma(t) = \frac{d}{dt}S[\rho_S(t)] - \beta \int dx_S H^*(\lambda_t) \frac{d}{dt}\rho_S(t)
\end{equation}
with $S[\rho_S(t)] \equiv -\int dx_S \rho(x_S;t)\ln\rho(x_S;t)$. The second term can be cast into the form 
\begin{equation}
 \int dx_S H^*(\lambda_t) \frac{d}{dt}\rho_S(t) = \frac{d}{dt}\lr{H^*(\lambda_t)} - \lr{\frac{dH^*(\lambda_t)}{dt}},
\end{equation}
where $\lr{\dots}$ denotes a phase space average with respect to $\rho_S(t)$. After realizing that 
$d_t H^*(\lambda_t) = d_t H(\lambda_t)$, we see that the last term coincides with the rate of work done on the system 
\begin{equation}\label{eq work rate Hamiltonian}
 \dot W(t) = \int dx_S \frac{dH(\lambda_t)}{dt}\rho_S(t).
\end{equation}
Using 
\begin{equation}
 \begin{split}
  \int dx_S \frac{dH(\lambda_t)}{dt}\rho_S(t)	&=	\int dx_{SB} \frac{dH_\text{tot}(\lambda_t)}{dt}\rho_\text{tot}(t)	\\
						&=	\int dx_{SB} \frac{d}{dt}[H_\text{tot}(\lambda_t)\rho_\text{tot}(t)],
 \end{split}
\end{equation}
this can be integrated to  
\begin{align}
 W(t)	&=	\int_0^t ds \lr{\frac{d H(\lambda_s)}{ds}}	\label{eq def work}	\\
	&=	\int dx_{SB} \left[H_\text{tot}(\lambda_t)\rho_\text{tot}(t) - H_\text{tot}(\lambda_0)\rho_\text{tot}(0)\right],	\nonumber
\end{align}
showing that the work done on the system is given by the total energetic change of the composite system and 
environment. The EP rate can then be expressed as 
\begin{equation}
 \dot\Sigma(t) = \beta\left[\dot W(t) - \frac{d}{dt}\lr{H^*(\lambda_t) + \frac{1}{\beta}\ln\rho_S(t)}\right].
\end{equation}
This motivates again the following definition of the non-equilibrium free energy 
[cf. Eq.~(\ref{eq free energy mesolevel})] 
\begin{equation}
 F(t) \equiv \lr{H^*(\lambda_t) + \frac{1}{\beta}\ln\rho_S(t)}	\label{eq free energy HMF}
\end{equation}
such that $\dot\Sigma(t) = \beta[\dot W(t) - d_t F(t)]$. 

For a useful thermodynamic framework, it now remains to show that the second law as known from phenomenological 
non-equilibrium thermodynamics holds: 
\begin{equation}\label{eq ent prod Seifert}
 \Sigma(t) \equiv \beta[W(t) - \Delta F(t)] \ge 0.
\end{equation}
For this purpose we assume as in the previous section that the initial state $\rho_S(0)$ belongs to the set 
$\C A_\pi$, see Eq.~(\ref{eq stationary preparation class}). The conditional equilibrium state of the 
bath is given by 
\begin{equation}
 \pi_{B|S} \equiv \frac{e^{-\beta(V + H_B)}}{\int dx_B e^{-\beta(V + H_B)}} = \frac{e^{-\beta[H_\text{tot}(\lambda_0) - H(\lambda_0)]}}{\C Z_B}.
\end{equation}
To prove the positivity of the EP, we refer to Ref.~\cite{SeifertPRL2016}, where it was deduced from an integral 
fluctuation theorem, or alternatively, the positivity becomes evident by noting the relation 
$\Sigma(t) = D[\rho_{SB}(t)\|\rho_S(t)\pi_{B|S}]$ and by recalling that the relative entropy is always 
positive~\cite{MillerAndersPRE2017, StrasbergEspositoPRE2017}. It is important to realize, however, that 
$\Sigma(t) \ge 0$ relies crucially on the choice of initial state. If $\rho_\text{tot}(0)\notin\C A_\pi$, we have 
\begin{equation}
 \begin{split}\label{eq ent prod Seifert arb IS}
  & \beta[W(t)-\Delta F(t)] = 	\\
  & D[\rho_{SB}(t)\|\rho_S(t)\pi_{B|S}] - D[\rho_{SB}(0)\|\rho_S(0)\pi_{B|S}],
 \end{split}
\end{equation}
which can be negative. 

After we have established that $\Sigma(t) = \int_0^t ds\dot\Sigma(s) \ge 0$ with the EP rate $\dot\Sigma(t)$ from 
Eq.~(\ref{eq ent prod abstract}), we can use the insights from Sec.~\ref{sec mathematical results} and 
Appendix~\ref{sec app Hamiltonian dynamics}. Then, we can immediately confirm the validity of the following theorem: 

\begin{thm}\label{thm ent prod HMF}
 Let $\rho_\text{tot}(0)\in\C A_\pi$ and let 
 $I$ denote the time interval in which the system dynamics is 1-Markovian and the process is 
 \begin{enumerate}
  \item undriven, or 
  \item driven and lumpable, or 
  \item driven and $\C A(t)\subset\C A_\pi$ or $\pi_\text{tot}(\lambda_t)\in\C A(t)$. 
 \end{enumerate}
 Then, $\dot\Sigma(t) \ge 0$ for all $t\in I$ and all admissible initial states. 
\end{thm}

We can therefore conclude for this setup that $\dot\Sigma(t) < 0$ directly implies non-Markovian dynamics for 
\emph{undriven} systems. For \emph{driven} systems this relation ceases to exist, but similar to Theorem~\ref{thm TSS} 
$\dot\Sigma(t) < 0$ implies that the two assumptions of 1-Markovian dynamics and a bath in a conditional equilibrium state 
cannot be simultaneously fulfilled. Two further remarks are in order: 

First, although it is possible to extend the framework of Ref.~\cite{SeifertPRL2016} to the situation of a time-dependent 
coupling Hamiltonian $V(\lambda_t)$ (see Ref.~\cite{StrasbergEspositoPRE2017}), Theorem~\ref{thm ent prod HMF} then 
ceases to hold because the work~(\ref{eq def work}) cannot anymore be computed from knowledge of the system state 
alone [also compare with Eq.~(\ref{eq work bipartite})]. 

Second, we remark that Theorem~\ref{thm ent prod HMF} is structurally identical to Theorem~\ref{thm ent prod bipartite case}. 
This shows the internal consisteny of our approach: since it is in principle possible to derive a ME from underlying 
Hamiltonian dynamics, we should find parallel results at each level of the description. This structural similarity was also 
found in Ref.~\cite{StrasbergEspositoPRE2017}. 

Also in parallel to Sec.~\ref{sec coarse-grained dissipative dynamics}, we remark that the splitting 
of the free energy $F = U-S/\beta$ does not allow to unambiguously define an internal energy and 
entropy. Hence, also the definition of heat via the first law $\Delta U = Q + W$ becomes 
ambiguous~\cite{TalknerHaenggiPRE2016}. However, the following definitions are appealing 
\begin{align}
 U(t)	&\equiv	\int dx_S \rho_S(t) \left[H^*(\lambda_t) + \beta \partial_\beta H^*(\lambda_t)\right],	\label{eq int energy HMF}	\\
 S(t)	&\equiv	\int dx_S \rho_S(t) \left[-\ln\rho_S(t) + \beta^2 \partial_\beta H^*(\lambda_t)\right],	\label{eq sys entropy HMF}
\end{align}
which can be shown to coincide (apart from a time-independent additive constant) with the global energy and entropy in 
equilibrium~\cite{SeifertPRL2016}. Further support for these definitions was given in 
Ref.~\cite{StrasbergEspositoPRE2017}, see also the discussion in Ref.~\cite{JarzynskiPRX2017}. 

Finally, to gain further insights into our approach, it is useful to reformulate it in terms of expressions which were 
previously derived for classical Hamiltonian dynamics~\cite{KawaiParrondoVandenBroeckPRL2007, 
VaikuntanathanJarzynskiEPL2009, HasegawaEtAlPLA2010, TakaraHasegawaDriebePLA2010, EspositoVandenBroeckEPL2011}. 
It follows from straightforward algebra that 
\begin{equation}\label{eq help 2}
 D[\rho_\text{tot}(t)\|\pi_\text{tot}(\lambda_t)] = \beta[F_\text{tot}(t) - \C F_\text{tot}(\lambda_t)],
\end{equation}
where $F_\text{tot}(t) = \lr{H_\text{tot}(\lambda_t)} + \lr{\ln\rho_\text{tot}(t)}/\beta$ is the non-equilibrium free 
energy associated to the global state $\rho_\text{tot}(t)$ and $\C F_\text{tot}(\lambda_t)$ is the equilibrium free 
energy associated to the thermal state $\pi_\text{tot}(\lambda_t)$. Due to Eq.~(\ref{eq help 2}) we can write the 
global EP rate as 
\begin{equation}
 \begin{split}\label{eq irreversible work}
  \dot\Sigma_\text{tot}(t)	&=	-\left.\frac{\partial}{\partial t}\right|_{\lambda_t} D[p_x(t)\|\pi_x(\lambda_t)]	\\
				&=	\beta\left[\dot W_\text{irr}(t) - \frac{d}{dt} D[p_x(t)\|\pi_x(\lambda_t)]\right] = 0,
 \end{split}
\end{equation}
which is zero for Hamiltonian dynamics. Here, $\dot W_\text{irr}(t) \equiv \dot W - d_t \C F_\text{tot}(\lambda_t)$ is 
the irreversible work and thus, Eq.~(\ref{eq irreversible work}) recovers (parts of) the earlier results from 
Refs.~\cite{KawaiParrondoVandenBroeckPRL2007, VaikuntanathanJarzynskiEPL2009, HasegawaEtAlPLA2010, 
TakaraHasegawaDriebePLA2010, EspositoVandenBroeckEPL2011}. Especially for an initially equilibrated microstate we 
immediately get the well-known dissipation inequality $W_\text{irr}(t) = D[p_x(t)\|\pi_x(\lambda_0)] \ge 0$. Now, from 
our findings above we see that we obtain an identical structure at the coarse-grained level: by using the 
identity~(\ref{eq help 2}) for the system, $D[\rho_S(t)\|\pi_S(\lambda_t)] = \beta[F(t) - \C F(\lambda_t)]$, we obtain 
\begin{equation}
 \dot\Sigma(t) = \beta\left[\dot W_\text{irr}(t) - \frac{d}{dt} D[p_\alpha(t)\|\pi_\alpha(\lambda_t)]\right].
\end{equation}
This expression can in general be negative and the conditions which ensure non-negativity are stated in 
Theorem~\ref{thm ent prod HMF}.

%%%%%%%%%%%%%%%%%%%%%%%%%%%%%%%%%%%%%%%%%%%%%%%%%%%%%%%%%%%%%%%%%%%%%%%%%%%%%%%%%%%%%%%%%%%%%%%%%%%%%%%%%%%%%%%%%%%%%%%%
\section{Strong coupling thermodynamics of quantum systems}
\label{sec thermo quantum}

So far we have only treated classical systems, but the question of how to obtain a meaningful thermodynamic description 
for quantum systems beyond the weak coupling and Markovian approximation is of equal importance. Whereas in 
Sec.~\ref{sec classical system bath theory} we could resort to an already well-developed framework, no general finite-time 
thermodynamic description for a driven quantum system immersed in an arbitrary single heat bath has been presented yet. 
Based on results obtained at equilibrium~\cite{GelinThossPRE2009, HsiangHuEntropy2018}, we first of all develop in 
Sec.~\ref{sec integrated description} the quantum extension of the framework introduced in Ref.~\cite{SeifertPRL2016}. 
Afterwards, in Sec.~\ref{sec breakdown} we prove that the relation worked out between non-Markovianity and 
a negative EP rate for classical systems cannot be established for quantum systems. The latter point is further 
studied in Sec.~\ref{sec quantum example} for the commonly used assumption that the system and bath are initially 
decorrelated; an assumption which is not true for the class of initial states considered in this section.

%%%%%%%%%%%%%%%%%%%%%%%%%%%%%%%%%%%%%%%%%%%%%%%%%%%%%%%%%%%%%%%%%%%%%%%%%%%%%%%%%%%%%%%%%%%%%%%%%%%%%%%%%%%%%%%%%%%%%%%%
\subsection{Integrated description}
\label{sec integrated description}

As in Sec.~\ref{sec classical system bath theory} our starting point is a time-dependent system-bath Hamiltonian 
of the form $\hat H_\text{tot}(\lambda_t) = \hat H(\lambda_t) + \hat V + \hat H_B$, where we used a hat to explicitly 
denote operators. The Hamiltonian of mean force in the quantum case is formally given by 
\begin{equation}
 \hat H^*(\lambda_t) = -\frac{1}{\beta}\ln\frac{\mbox{tr}_B\{e^{-\beta[\hat H(\lambda_t) + \hat V + \hat H_B]}\}}{Z_B}
\end{equation}
and it shares the same meaning as in the classical case, cf.~Eq.~(\ref{eq HMF}): it describes the exact reduced state of 
the system if the system-bath composite is in a global equilibrium state. Motivated by equilibrium considerations and by 
Sec.~\ref{sec classical system bath theory}, we define the three key thermodynamic quantities internal energy, system 
entropy and free energy for an arbitrary system state $\hat\rho_S(t)$ as follows: 
\begin{align}
 U(t)	&\equiv	\mbox{tr}_S\left\{\hat\rho_S(t)\left[\hat H^*(\lambda_t) + \beta\partial_\beta\hat H^*(\lambda_t)\right]\right\}, \label{eq def U quantum}	\\
 S(t)	&\equiv	\mbox{tr}_S\left\{\hat\rho_S(t)\left[-\ln \hat\rho_S(t) + \beta^2\partial_\beta\hat H^*(\lambda_t)\right]\right\}, \label{eq def S quantum}	\\
 F(t)	&\equiv	\mbox{tr}_S\left\{\hat\rho_S(t)\left[\hat H^*(\lambda_t) + \frac{1}{\beta}\ln \hat\rho_S(t)\right]\right\}. \label{eq def F quantum}
\end{align}
Note that all quantities are state functions. Also the definition of work is formally identical to 
Sec.~\ref{sec classical system bath theory}, Eq.~(\ref{eq def work}), 
\begin{align}
 W(t)	&=	\int_0^t ds \mbox{tr}_S\left\{\frac{d\hat H(\lambda_s)}{ds}\hat \rho_S(s)\right\}	\label{eq def work quantum}	\\
	&=	\mbox{tr}_{SB}\{\hat\rho_\text{tot}(t)\hat H_\text{tot}(\lambda_t)\} - \mbox{tr}_{SB}\{\hat\rho_\text{tot}(0)\hat H_\text{tot}(\lambda_0)\},	\nonumber
\end{align}
and the heat flux is again fixed by the first law $Q(t) = \Delta U(t) - W(t)$. 

Equipped with these definitions, we define the EP 
\begin{equation}\label{eq 2nd law quantum}
 \Sigma(t) \equiv \beta[W(t) - \Delta F(t)]
\end{equation}
as usual and ask when can we ensure its positivity? Again, in complete analogy to Eq.~(\ref{eq ent prod Seifert arb IS}) 
one can show that 
\begin{equation}
 \begin{split}\label{eq ent prod quantum arb IS}
  & \beta[W(t) - \Delta F(t)] = \\
  & D\left[\hat\rho_\text{tot}(t)\left\|\hat\pi_\text{tot}(\lambda_t)\right]\right. - D\left[\hat\rho_{S}(t)\left\|\hat\pi_S(\lambda_t)\right]\right.	\\
  & -D\left[\hat\rho_\text{tot}(0)\left\|\hat\pi_\text{tot}(\lambda_0)\right]\right. + D\left[\hat\rho_{S}(0)\left\|\hat\pi_S(\lambda_0)\right]\right.,
 \end{split}
\end{equation}
where $D[\hat\rho\|\hat\sigma] \equiv \mbox{tr}\{\hat\rho(\ln\hat\rho-\ln\hat\sigma)\}$ is the quantum relative 
entropy and $\hat\pi_\text{tot}(\lambda_t)$ the global Gibbs state and 
$\hat\pi_S(\lambda_t) = \mbox{tr}_B\{\hat\pi_\text{tot}(\lambda_t)\}$. 
Eq.~(\ref{eq ent prod quantum arb IS}) can be derived by using that the von Neumann entropy of the global state 
$S[\hat\rho_\text{tot}(t)] = -\mbox{tr}_{SB}\{\hat\rho_\text{tot}(t)\ln\hat\rho_\text{tot}(t)\}$ is conserved and by 
using the relation 
$\ln[\frac{\C Z^*(\lambda_t)}{\C Z_\text{tot}(\lambda_t)}\frac{\C Z_\text{tot}(\lambda_0)}{\C Z^*(\lambda_0)}] = \ln\frac{\C Z_B}{\C Z_B} = 0$, 
where the partition functions are defined analogously to Eq.~(\ref{eq HMF}). Notice that this identity requires the bath 
Hamiltonian to be undriven. 

We now note that due to the monotonicity of relative entropy~\cite{UhlmannCMP1977, OhyaPetzBook1993} the first line in 
Eq.~(\ref{eq ent prod quantum arb IS}) is never negative, while the second line is never positive. Hence, positivity 
of the EP~(\ref{eq 2nd law quantum}) is ensured if 
\begin{equation}\label{eq cond positivity 2nd law}
 D\left[\hat\rho_\text{tot}(0)\left\|\hat\pi_\text{tot}(\lambda_0)\right]\right. - D\left[\hat\rho_{S}(0)\left\|\hat\pi_S(\lambda_0)\right]\right. = 0.
\end{equation}
Two important classes of initial states for which this is the case are: 

\emph{Class 1 (global Gibbs state).} If the initial composite system-bath state is a Gibbs state 
$\hat\pi_\text{tot}(\lambda_0)$, we immediately see that Eq.~(\ref{eq cond positivity 2nd law}) is fulfilled and 
$\beta[W(t) - \Delta F(t)] \ge 0$ holds true. For a cyclic process, in which the system Hamiltonian is the same at the 
initial and final time, positivity of Eq.~(\ref{eq 2nd law quantum}) follows alternatively from the approach in 
Ref.~\cite{UzdinSaarPRX2018}. 

\emph{Class 2 (commuting initial state).} We consider initial states of the form 
\begin{equation}\label{eq commuting initial state}
 \hat\rho_\text{tot}(0) = \sum_k p_k(0)\hat\Pi_k \hat\rho_{B|k}(\lambda_0),
\end{equation}
where the $\hat\Pi_k = |k\rangle\langle k|$ are orthogonal rank-1 projectors in the system space fulfilling the 
commutation relations 
\begin{equation}\label{eq commuting condition}
 [\hat\Pi_k,\hat H^*(\lambda_0)] = [\hat\Pi_k,\hat H_\text{tot}(\lambda_0)] = 0 ~ \forall k.
\end{equation}
This is ensured when $[\hat H(\lambda_0),\hat V] = 0$. 
The state of the bath conditioned on the system state $\hat\Pi_k$ reads 
\begin{equation}\label{eq cond bath state quantum}
 \hat\rho_{B|k}(\lambda_0) = \frac{\mbox{tr}_S\{\hat\Pi_k\hat\pi_\text{tot}(\lambda_0)\}}{\mbox{tr}_{SB}\{\hat\Pi_k\hat\pi_\text{tot}(\lambda_0)\}} = \frac{\lr{k|\hat\pi_\text{tot}(\lambda_0)|k}}{\lr{k|\hat\pi_S(\lambda_0)|k}}.
\end{equation}
Since the $p_k(0)$ are allowed to be arbitrary probabilities, Eq.~(\ref{eq commuting initial state}) is the direct 
quantum analogue of the initial states considered in the classical setting in Sec.~\ref{sec classical system bath theory}. 
Using condition~(\ref{eq commuting condition}) it becomes a task of straightforward algebra to show that 
Eq.~(\ref{eq cond positivity 2nd law}) holds. 

We remark that all considerations above can be also extended to a time-dependent coupling Hamiltonian, i.e., by allowing 
$\hat V = \hat V(\lambda_t)$ to depend on time. Again, the problem is then that the work~(\ref{eq def work quantum}) 
cannot be computed based on the knowledge of the system state $\hat\rho_S(t)$ alone. Furthermore, it is worth to point 
out that positivity of the second law~(\ref{eq 2nd law quantum}) with the \emph{nonequilibrium} free energy represents 
a stronger inequality than the bound for the dissipated work derived in Ref.~\cite{CampisiTalknerHaenggiPRL2009}
from a fluctuation theorem using the equilibrium free energy.

%%%%%%%%%%%%%%%%%%%%%%%%%%%%%%%%%%%%%%%%%%%%%%%%%%%%%%%%%%%%%%%%%%%%%%%%%%%%%%%%%%%%%%%%%%%%%%%%%%%%%%%%%%%%%%%%%%%%%%%%
\subsection{Breakdown of the results from Sec.~\ref{sec classical system bath theory}}
\label{sec breakdown}

The positivity of $\Sigma(t)$ could be established for initial global Gibbs states or for commuting initial states. 
Without any driving ($\dot\lambda_t = 0$) these states are not very interesting as they remain invariant in time. Hence, 
we only consider the driven situation. Clearly, the analogue of Eq.~(\ref{eq 2nd law quantum}) at the rate level is 
$\beta[\dot W(t) - d_t F(t)]$. Unfortunately, this does not coincide with the quantum 
counterpart of Eq.~(\ref{eq ent prod abstract}). To see this, suppose that 
\begin{equation}
 \dot\Sigma(t) = -\left.\frac{\partial}{\partial t}\right|_{\lambda_t} D[\hat\rho_S(t)\|\hat\pi_S(\lambda_t)].
\end{equation}
This can be rewritten as 
\begin{equation}
 \begin{split}
  \dot\Sigma(t) =&~	\frac{d}{dt}\left\{ S[\hat\rho_S(t)] - \beta \lr{\hat H^*(\lambda_t)}\right\}	\\
		  &+	\beta \mbox{tr}\left\{\hat\rho_S(t)\frac{d\hat H^*(\lambda_t)}{dt}\right\}.
 \end{split}
\end{equation}
Unfortunately, the analogy with Sec.~\ref{sec classical system bath theory} stops here because the last term cannot 
be identified with the work done on the quantum system and hence, $\int_0^t ds\dot\Sigma(s) \neq \Sigma(t)$. 
In fact, 
\begin{equation}
 \frac{\partial\hat H^*(\lambda_t)}{\partial t} \neq \frac{\partial\hat H(\lambda_t)}{\partial t}
\end{equation}
unless in the ``classical'' (and for us uninteresting) limit $[H(\lambda_t),V] = 0$. 

To conclude, for quantum systems the EP rate cannot be expressed in terms of a relative entropy describing the 
irreversible relaxation to the equilibrium state, which would be desirable because an analogue of 
Lemma~\ref{lemma Markov contractivity} holds also in the quantum case~\cite{SpohnJMP1978}. Thus, the very existence of 
a general relation between EP and non-Markovianity as established for previous setups seems questionable at the moment. 
This conclusion can be drawn without touching upon the difficult question of how to extend many of the mathematical 
results of Sec.~\ref{sec mathematical results} to the quantum case.

%%%%%%%%%%%%%%%%%%%%%%%%%%%%%%%%%%%%%%%%%%%%%%%%%%%%%%%%%%%%%%%%%%%%%%%%%%%%%%%%%%%%%%%%%%%%%%%%%%%%%%%%%%%%%%%%%%%%%%%%
\section{Applications}
\label{sec applications}

After having established the general theory in the last four sections, we now consider various examples and 
applications. However, it is not our intention here to cover every aspect of our theory. We rather prefer to focus on 
simple models, whose essence is easy to grasp and which illuminate certain key aspects of our framework, thereby also 
shedding light on some misleading statements made in the literature.

%%%%%%%%%%%%%%%%%%%%%%%%%%%%%%%%%%%%%%%%%%%%%%%%%%%%%%%%%%%%%%%%%%%%%%%%%%%%%%%%%%%%%%%%%%%%%%%%%%%%%%%%%%%%%%%%%%%%%%%%
\subsection{Time-dependent instantaneous fixed points for an undriven ergodic Markov chain}
\label{sec steady states}

For the formal development of our theory it was of crucial importance to know under which conditions we could ensure 
that there is a well-defined IFP $\pi_\alpha(\lambda_t)$ for the coarse-grained dynamics, which follows from an 
underlying steady state of the microdynamics. Especially for driven systems this was hard to establish because even 
when we start with the initial steady state $\pi_{x}(\lambda_0)$, the driving will take it out of that state 
such that $p_x(t) \neq \pi_x(\lambda_t)$ in general. One might wonder whether additional conditions, such as 
1-Markovianity or ergodicity, help to ensure that $\pi_\alpha(\lambda_t)$ is an IFP of the mesodynamics, but we will 
here show that this is not the case. 

As a counterexample we consider a simple three-state system described by a three-by-three rate matrix $W(\lambda_t)$. 
Imagine that the system started in $\C A(0) \subset\C A_\pi(\lambda_0)$, i.e., the initial microstates were conditionally 
equilibrated. The system is then subjected to an arbitrary driving protocol $\lambda_t$ up to some time $t^*$. Afterwards, 
we keep the protocol fixed, i.e., $\lambda_t = \lambda_{t^*}$ for all $t\ge t^*$. Clearly, at time $t^*$ the microstates 
will in general not be conditionally equilibrated, i.e., $\C A(t)\nsubseteq \C A_\pi(\lambda_{t^*})$. 

Now, for definiteness we choose the full rate matrix describing the evolution of the probability vector 
$\bb p(t) = [p_1(t),p_2(t),p_3(t)]$ for $t\ge t^*$ to be 
\begin{equation}
 W(\lambda_{t^*}) = \left(\begin{array}{ccc}
            -1-e^{-\epsilon/2}	&	1			&	e^{\epsilon/2} \\
	    1			&	-1-e^{-\epsilon/2}	&	e^{\epsilon/2} \\
            e^{-\epsilon/2}	&	e^{-\epsilon/2}		&	-2 e^{\epsilon/2} \\
           \end{array}\right).
\end{equation}
It obeys local detailed balance~(\ref{eq local detailed balance}) if we parameterize the inverse temperature and 
energies as $\beta E_1 = \beta E_2 = 0$ and $\beta E_3 = \epsilon$ and furthermore we have set any kinetic 
coefficients in the rates equal to one. As a partition we choose $\chi_\alpha = \{1\}$ and $\chi_{\alpha'} = \{2,3\}$ and 
in the long time limit the mesostates will thermalize appropriately for any initial state, 
\begin{equation}\label{eq example 1 inststst}
 \binom{\pi_\alpha}{\pi_{\alpha'}} = \lim_{t\rightarrow\infty}\binom{p_\alpha(t)}{p_{\alpha'}(t)} = \frac{1}{e^{-\epsilon}+2}\binom{1}{1+e^{-\epsilon}},
\end{equation}
i.e., the rate matrix $W(\lambda_{t^*})$ is \emph{ergodic}. 

As emphasized above, the conditional microstates need not be in equilibrium initially and we parametrize them by 
$p_{2|\alpha'}(t^*) = \gamma$, $p_{3|\alpha'}(t^*) = 1-\gamma$ ($\gamma\in[0,1]$). In principle it is possible to 
analytically compute the generator~(\ref{eq meso generator ME}) for the ME at the mesolevel, but we refrain from showing 
the resulting very long expression. Instead, we focus on Fig.~\ref{fig plot ex 1}. It clearly shows that the IFP of the 
dynamics is given by Eq.~(\ref{eq example 1 inststst}) only if we choose 
$p_{2|\alpha'}(t^*) = \pi_{2|\alpha'}(\lambda_{t^*})$ and $p_{3|\alpha'}(t^*) = \pi_{3|\alpha'}(\lambda_{t^*})$ 
[implying $\gamma = \gamma_\text{eq} \equiv e^\epsilon/(1+e^\epsilon)$], 
i.e., if the microstates are conditionally equilibrated in agreement with Theorem~\ref{thm steady state weak}. We have 
also checked that the time-dependent rates of the generator~(\ref{eq meso generator ME}) are always positive for this 
example (not shown here for brevity) and hence, the dynamics is 1-Markovian. 

\begin{figure}%[h]
 \centering\includegraphics[width=0.44\textwidth,clip=true]{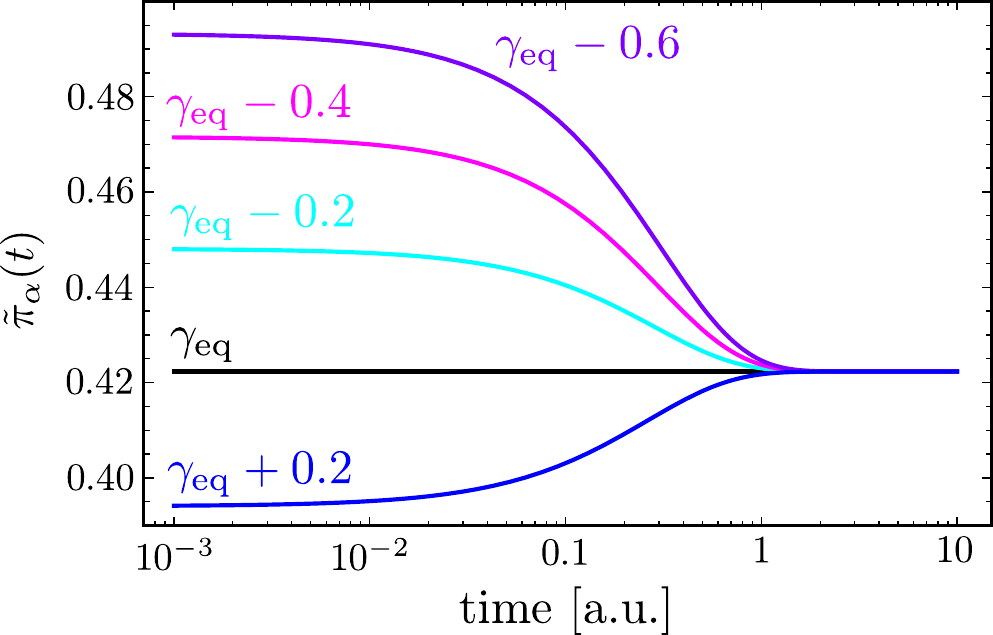}
 \label{fig plot ex 1} 
 \caption{Plot of the changing IFP, denoted here by $\tilde\pi_\alpha(t)$, over time $t$ in logarithmic scale (for the 
 plot we set the initial time $t^* = 0$). The figure shows that $\tilde\pi_\alpha(t)\neq \pi_\alpha(\lambda_{t^*})$ 
 unless we choose $\gamma = \gamma_\text{eq}$. In the long-time limit the IFP coincides with the equilibrium 
 distribution~(\ref{eq example 1 inststst}). We set $\epsilon = 1$ which implies $\gamma_\text{eq} \approx 0.73$. }
\end{figure}

This example proves that ergodicity does not imply that $\pi_\alpha(\lambda_t)$ is the IFP of the reduced dynamics, 
as claimed in Ref.~\cite{SpeckSeifertJSM2007} for arbitrary non-Markovian dynamics. Even 1-Markovianity together with 
ergodicity is not sufficient to ensure this statement.

%%%%%%%%%%%%%%%%%%%%%%%%%%%%%%%%%%%%%%%%%%%%%%%%%%%%%%%%%%%%%%%%%%%%%%%%%%%%%%%%%%%%%%%%%%%%%%%%%%%%%%%%%%%%%%%%%%%%%%%%
\subsection{Markovianity without time-scale separation}
\label{sec strong lumpability without TSS}

We give a simple example of a physically relevant and lumpable Markov process although TSS does not apply. For this 
purpose consider the following rate matrix 
\begin{equation}\label{eq spin valve rate matrix}
 W = 
 \left(\begin{array}{ccc}
        -2\gamma_\text{in}	&	\gamma_\text{out}				&	\gamma_\text{out}	\\
        \gamma_\text{in}	&	-\gamma_\text{out}-\bar\gamma_\text{flip}	&	\gamma_\text{flip}			\\
        \gamma_\text{in}	&	\bar\gamma_\text{flip}				&	-\gamma_\text{out}-\gamma_\text{flip}	\\
       \end{array}\right)
\end{equation}
describing the time evolution of a probability vector $\bb p(t) = [p_0(t),p_\uparrow(t),p_\downarrow(t)]$. 
This ME describes a quantum dot in the ultrastrong Coulomb blockade regime coupled to a metallic lead taking 
the spin degree of freedom into account. Then, $p_{0/\uparrow/\downarrow}(t)$ are the probabilities to find 
the dot at time $t$ in a state with zero electrons, an electron with spin up or an electron with spin down, 
respectively. If the metallic lead has a finite magnetization, the rates for hopping in ($\gamma_\text{in}$) 
and out ($\gamma_\text{out}$) of the quantum depend on the spin, which can be derived from first 
principles~\cite{BraunKoenigMartinekPRB2004} and has interesting thermodynamic 
applications~\cite{StrasbergEtAlPRE2014}. But if the lead has zero magnetization as considered here, the 
dynamics of the spin degree of freedom do not matter. Hence, if we consider the partition $\chi_0=\{0\}$ and 
$\chi_1=\{\uparrow,\downarrow\}$, it is not hard to deduce that 
\begin{equation}\label{eq spin valve ME}
 \frac{\partial}{\partial t}\binom{p_0(t)}{p_1(t)} = 
 \left(\begin{array}{cc}
        -2\gamma_\text{in}	&	\gamma_\text{out}	\\
        2\gamma_\text{in}	&	-\gamma_\text{out}	\\
       \end{array}\right) \binom{p_0(t)}{p_1(t)}
\end{equation}
where $p_1(t) = p_\uparrow(t) + p_\downarrow(t)$. Thus, the coarse-grained dynamics is Markovian 
for all times $t$ and all micro initial conditions $[p_0(0),p_\uparrow(0),p_\downarrow(0)]$ although TSS does not apply. 
Notice that the IFP of Eq.~(\ref{eq spin valve ME}) coincides with the marginalized IFP of 
Eq.~(\ref{eq spin valve rate matrix}) and hence, we have $\dot\Sigma(t) \ge 0$. Moreover, as long as the structure of 
the rate matrix~(\ref{eq spin valve rate matrix}) is preserved, we could have even allowed for arbitrary 
time-dependencies in the rates.

%%%%%%%%%%%%%%%%%%%%%%%%%%%%%%%%%%%%%%%%%%%%%%%%%%%%%%%%%%%%%%%%%%%%%%%%%%%%%%%%%%%%%%%%%%%%%%%%%%%%%%%%%%%%%%%%%%%%%%%%
\subsection{Classical Brownian motion}
\label{sec Brownian motion}

We here present an example which exhibits negative EP rates and link their appearance to the spectral features of 
the environment. This is done by considering the important class of driven, classical Brownian motion models (also 
called Caldeira-Leggett or independent oscillator models). The global Hamiltonian with mass-weighted coordinates reads 
\begin{align}\label{eq Brownian motion Hamiltonian}
 H(\lambda_t)	&=	\frac{1}{2}[p^2 + \omega^2(\lambda_t)x^2],	\\
 V + H_B	&=	\frac{1}{2}\sum_k\left[p_k^2 + \nu_k^2\left(x_k - \frac{c_k}{\nu_k^2}x\right)^2\right],
\end{align}
and its study has attracted considerable interest in strong coupling thermodynamics~\cite{MartinezPazPRl2013, 
PucciEspositoPelitiJSM2013, StrasbergEtAlNJP2016, FreitasPazPRE2017, StrasbergEspositoPRE2017, AurellEnt2017, 
PerarnauLlobetEtAlPRL2018, HsiangEtAlPRE2018}. 
The Hamiltonian describes a central oscillator with position $x$ and momentum $p$ linearly coupled to a set of bath 
oscillators with positions $x_k$ and momenta $p_k$. The frequency of the central oscillator can be driven and we 
parametrize it as $\omega(\lambda_t) = \omega_0 + g\sin(\omega_Lt)$. Furthermore, $c_k$ and $\nu_k$ are the 
system-bath coupling constants and the frequencies of the bath oscillators. It turns out that all the information 
about the bath (except of its temperature) can be encoded into a single function known as the spectral density of the 
bath. It is defined in general as $J(\omega) \equiv \frac{\pi}{2}\sum_k\frac{c_k^2}{\nu_k}\delta(\omega-\nu_k)$ 
and we parametrize it as 
\begin{equation}\label{eq SD non-Markovian}
 J(\omega) = \frac{\lambda_0^2\gamma\omega}{(\omega^2-\omega_1^2)^2 + \gamma^2\omega^2}.
\end{equation}
Here, $\lambda_0$ controls the overall coupling strength between the system and the bath and $\gamma$ changes the 
shape of the SD from a pronounced peak around $\omega_1$ for small $\gamma$ to a rather unstructured and flat SD for 
large $\gamma$. Thus, intuitively one expects that a smaller $\gamma$ corresponds to stronger non-Markovianity 
although this intuition can be misleading too~\cite{StrasbergEspositoPRL2018}. 

\begin{figure}%[h]
 \centering\includegraphics[width=0.47\textwidth,clip=true]{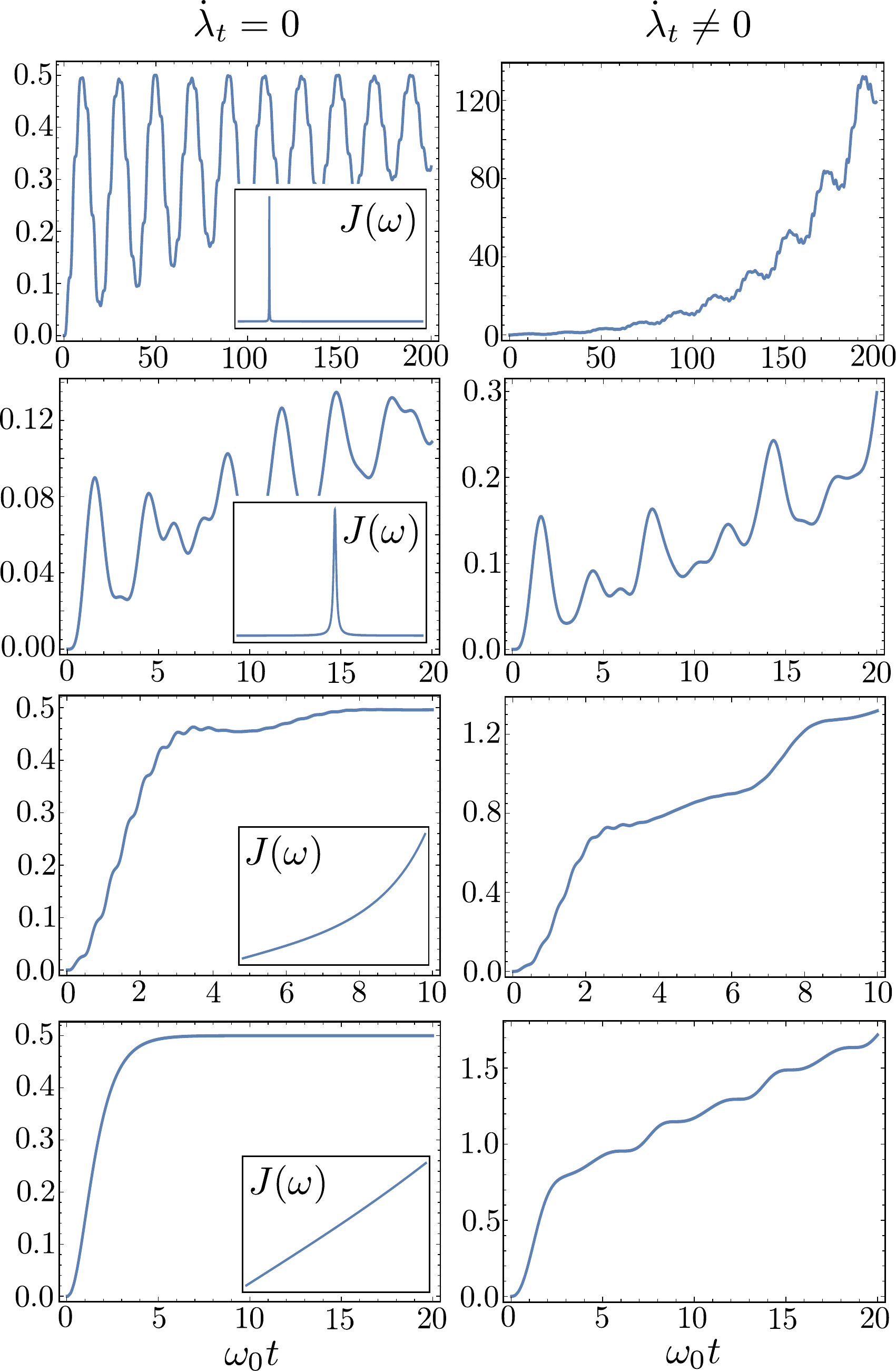}
 \label{fig plot ex 3} 
 \caption{Plot of the dimensionless entropy production $\Sigma(t)$ ($k_B\equiv 1$) over the dimensionless time 
 $\omega_0 t$ for different parameters. For the driving we chose $g = 0$ and $g = 0.3\omega_0$ for the left or right 
 column, respectively, and $\omega_L = \omega_0$. We changed the shape of the spectral density $J(\omega)$ in each row, 
 which is depicted for $\omega\in[0,6\omega_0]$ as a small inset (note that the vertical scaling is different in each 
 inset). Specifically, the parameters $(\lambda_0,\gamma,\omega_1)$ are $(0.316\omega_0,0.01,1)\omega_0$ (top), 
 $(3.16\omega_0,0.1,3.16)\omega_0$ (second row),  $(100\omega_0,1,10)\omega_0$ (third row), 
 $(500\omega_0,10,31.6)\omega_0$ (bottom). The system was prepared according to 
 Eq.~(\ref{eq stationary preparation class}) with initial mean values $\lr{x}(0) = (\sqrt{\beta}\omega_0)^{-1}$, 
 $\lr{p_x}(0) = 0$ and covariances $C_{xx}(0) = (\beta\omega_0^2)^{-1}, C_{p_xp_x}(0) = \beta^{-1}$ and 
 $C_{xp_x}(0) = 0$. Note that this specific choice corresponds to equilibrated covariances, but the mean values are 
 out of equilibrium. The general features of the plot, however, do not change too much for different non-equilibrium 
 initial states. Finally, we set $\omega_0 = 1$ and $\beta = 1$. See also Ref.~\cite{StrasbergEspositoPRE2017} for 
 details of the computation. }
\end{figure}

The dynamics of the model is exactly described by the generalized Langevin equation (see, e.g.,~\cite{WeissBook2008}) 
\begin{equation}\label{eq Langevin eq general}
 \ddot x(t) + \omega_0^2(t) x(t) + \int_0^t ds \Gamma(t-s)\dot x(s) = \xi(t)
\end{equation}
with the friction kernel 
\begin{equation}
 \Gamma(t) \equiv \int_0^\infty d\omega \frac{2}{\pi\omega} J(\omega) \cos(\omega t)
\end{equation}
and the noise $\xi(t)$, which -- when averaged over the initial state of the bath -- obeys the statistics 
\begin{equation}
 \lr{\xi(t)}_B = 0, ~~~ \lr{\xi(t)\xi(s)}_B = \frac{1}{\beta}\Gamma(t-s).
\end{equation}
To compute the thermodynamic quantities introduced in Sec.~\ref{sec classical system bath theory} we need the state of 
the system $\rho_S(t)$. It can be computed with the method explained in Sec.~IV of Ref.~\cite{StrasbergEspositoPRE2017}, 
which we will not repeat here. Instead, we focus on the explanation of the numerical observations only. 

Fig.~\ref{fig plot ex 3} gives illustrative examples of the time-evolution of the EP $\Sigma(t) \ge 0$ defined in 
Eq.~(\ref{eq ent prod Seifert}) for various situations. In total, we plot it for four different parameters 
characterizing the spectral density, always for the same initial condition of the system, but for the case of an 
undriven (left column) or a driven (right column) process. The parameters are chosen from top to bottom such that 
the spectral density resembles more and more an Ohmic spectral density $J(\omega) \sim \omega$, which usually gives 
rise to Markovian behaviour. In fact, this standard intuition is nicely confirmed in Fig.~\ref{fig plot ex 3} by 
observing that negative EP rates are much larger and much more common at the top. The plot at bottom indeed corresponds 
to the Markovian limit in which the bath is conditionally equilibrated throughout (this is similar to the limit 
of TSS treated in Sec.~\ref{sec time scale separation}, see also Ref.~\cite{StrasbergEspositoPRE2017} for additonal 
details). It is worthwhile to repeat that a negative EP rate in the left column of Fig.~\ref{fig plot ex 3} 
indicates non-Markovian behaviour in a rigorous sense, whereas for the right column this is only true in a weaker sense, 
but it unambiguously shows that the bath cannot be adiabatically eliminated.

%%%%%%%%%%%%%%%%%%%%%%%%%%%%%%%%%%%%%%%%%%%%%%%%%%%%%%%%%%%%%%%%%%%%%%%%%%%%%%%%%%%%%%%%%%%%%%%%%%%%%%%%%%%%%%%%%%%%%%%%
\subsection{Quantum dynamics under the initial product state assumption}
\label{sec quantum example}

We have shown in Sec.~\ref{sec thermo quantum} that the definition~(\ref{eq ent prod abstract}) of the EP rate for 
classical systems does not properly generalize to the quantum case. Part of the problem could be that we started from an 
initially correlated state, which complicates the treatment of the dynamics of the quantum system significantly. 
Therefore, one often resorts to the initial product state assumption 
$\hat\rho_\text{tot}(0) = \hat\rho_S(t) \otimes \hat\rho_B$, where $\hat\rho_S(t)$ is arbitrary and $\hat\rho_B$ fixed 
(usually taken to be the Gibbs state of the bath)~\cite{RivasHuelgaPlenioRPP2014, BreuerEtAlRMP2016, 
BreuerPetruccioneBook2002, DeVegaAlonsoRMP2017, EspositoLindenbergVandenBroeckNJP2010}. It is then interesting to ask 
which general statements connecting Markovianity, the notion of an IFP and EP rates can be made in this case. 
The following simple example shows which statements do \emph{not} hold in this case. 

A single fermionic mode (such as a quantum dot in the Coulomb blockade regime) tunnel-coupled to a bath of free fermions 
(describing, e.g., a metallic lead) can be modeled by the single resonant level Hamiltonian (assuming spin polarization) 
\begin{equation}
 \hat H_\text{tot} = \epsilon_0\hat d^\dagger\hat d + \sum_k \left(t_k\hat d\hat c_k^\dagger + t_k^*\hat c_k\hat d^\dagger + \epsilon_k\hat c_k^\dagger\hat c_k\right).
\end{equation}
Here, $\hat d^{(\dagger)}$ and $\hat c_k^{(\dagger)}$ are fermionic annihilation (creation) operators, $\epsilon_0$ is 
the real-valued energy of the quantum dot, $t_k$ is a complex tunnel amplitude and $\epsilon_k$ is the real-valued 
energy of a bath fermion. 

To describe the dynamics of the open system we use the Redfield ME~\cite{BreuerPetruccioneBook2002, DeVegaAlonsoRMP2017} 
\begin{align}
 \frac{\partial}{\partial t}\hat\rho_S(t)	&=	-i[\hat H,\hat\rho_S(t)]	\label{eq Redfield ME}	\\
						&	-\int_0^t ds\mbox{tr}_B\left\{[\hat V,[\hat V(s-t),\hat\rho_S(t)\otimes\hat\pi_B]]\right\}.	\nonumber
\end{align}
Here, the system and interaction Hamiltonian are $\hat H = \epsilon_0\hat d^\dagger\hat d$ and 
$\hat V = \sum_k(t_k\hat d\hat c_k^\dagger + t_k^* \hat c_k\hat d^\dagger)$. Furthermore, 
$\hat V(t) = e^{i(\hat H+\hat H_B)t/\hbar}\hat Ve^{-i(\hat H+\hat H_B)t/\hbar}$ denotes the interaction picture with 
$\hat H_B = \sum_k \epsilon_k \hat c_k^\dagger \hat c_k$. We assumed the initial 
system-bath state to be $\hat\rho_S(0)\otimes\hat\pi_B$ where $\hat\rho_S(0)$ is arbitrary and $\hat\pi_B$ the 
grand-canonical equilibrium state with respect to $\hat H_B$ and the particle number operator 
$\hat N_B = \sum_k \hat c_k^\dagger\hat c_k$. Without loss of generality we set the chemical potential to zero 
($\mu = 0$). The Redfield equation~(\ref{eq Redfield ME}) directly results from 
a perturbative expansion of the exact time-convolutionless ME and it gives accurate results for sufficiently 
small tunneling amplitudes $t_k$ and a relatively high bath temperature. 

Following standard procedures, we rewrite Eq.~(\ref{eq Redfield ME}) as 
\begin{align}
 \frac{\partial}{\partial t}\hat\rho_S(t)	=&	-i\epsilon(t)[\hat d^\dagger\hat d,\hat\rho_S(t)]	\\
						&	+\gamma_\text{out}(t)\left(\hat d\hat\rho_S(t)\hat d^\dagger - \frac{1}{2}\{\hat d^\dagger\hat d,\hat\rho_S(t)\}\right)	\nonumber	\\
						&	+\gamma_\text{in}(t)\left(\hat d^\dagger\hat\rho_S(t)\hat d - \frac{1}{2}\{\hat d\hat d^\dagger,\hat\rho_S(t)\}\right),	\nonumber
\end{align}
where $\{\cdot,\cdot\}$ denotes the anti-commutator and 
$\epsilon(t) \equiv \epsilon_0-\Delta_\text{in}(t)-\Delta_\text{out}(t)$ is a time-dependent renormalized system energy.
In detail, we have introduced the quantities 
\begin{align}
 \gamma_\text{in}(t)	&\equiv	\int_0^t d\tau\int_{-\infty}^\infty d\omega \frac{J(\omega)}{\pi}f(\omega) \cos[(\omega-\epsilon_0)\tau],	\label{eq rate in}	\\
 \Delta_\text{in}(t)	&\equiv	\int_0^t d\tau\int_{-\infty}^\infty d\omega \frac{J(\omega)}{2\pi}f(\omega) \sin[(\omega-\epsilon_0)\tau],	\\
 \gamma_\text{out}(t)	&\equiv	\int_0^t d\tau\int_{-\infty}^\infty d\omega \frac{J(\omega)}{\pi}[1-f(\omega)] \cos[(\omega-\epsilon_0)\tau],	\label{eq rate out}	\\
 \Delta_\text{out}(t)	&\equiv	\int_0^t d\tau\int_{-\infty}^\infty d\omega \frac{J(\omega)}{2\pi}[1-f(\omega)] \sin[(\omega-\epsilon_0)\tau],
\end{align}
where $f(\omega) \equiv (e^{\beta\omega}+1)^{-1}$ denotes the Fermi function for $\mu = 0$ and 
$J(\omega) \equiv 2\pi \sum_k |t_k|^2\delta(\omega-\epsilon_k)$ is the spectral density of the bath. 
If there are no initial coherences in the quantum dot present, we can conclude without any further approximation 
that the full dynamics of the quantum dot is captured by the rate ME 
\begin{equation}\label{eq ME SRL}
 \frac{\partial}{\partial t}\binom{p_1(t)}{p_0(t)} = 
 \left(\begin{array}{cc}
        -\gamma_\text{out}(t)	&	\gamma_\text{in}(t)	\\
        \gamma_\text{out}(t)	&	-\gamma_\text{in}(t)	\\
       \end{array}\right)\binom{p_1(t)}{p_0(t)},
\end{equation}
where $p_1(t)$ [$p_0(t)$] describes the probability to find the dot in the filled [empty] state at time $t$. 

\begin{figure}%[h]
 \centering\includegraphics[width=0.40\textwidth,clip=true]{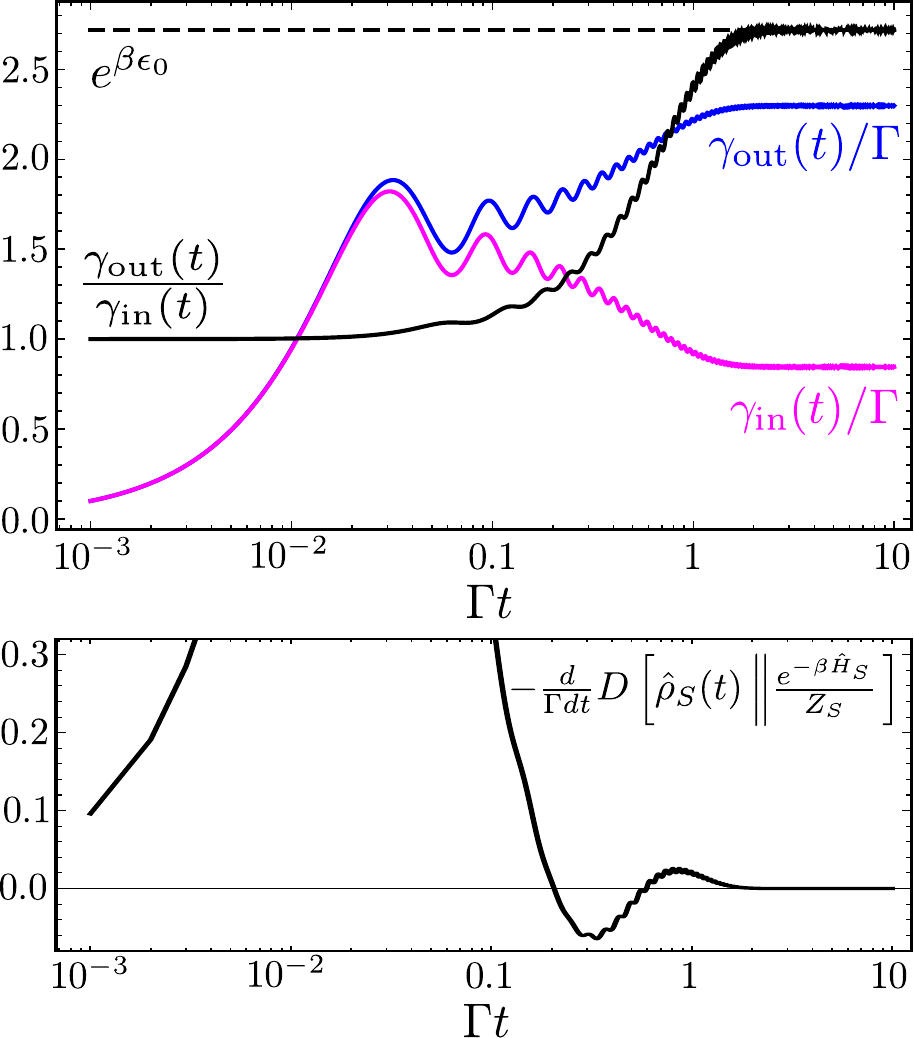}
 \label{fig plot ex 2} 
 \caption{\bb{Top:} Plot of the (dimensionless) rates $\gamma_\text{in}(t)/\Gamma$ and $\gamma_\text{out}(t)/\Gamma$ 
 defined in Eqs.~(\ref{eq rate in}) and~(\ref{eq rate out}), their ratio $\gamma_\text{out}(t)/\gamma_\text{in}(t)$ and 
 the expected local detailed balance ratio $e^{\beta\epsilon_0}$ over dimensionless time $\Gamma t$ in logarithmic scale. 
 For the plot we parametrized the bath spectral density as $J(\omega) = \Gamma$ for $\omega/\Gamma\in[-100,+100]$ and 
 zero outside. The dot energy and inverse temperature of the bath are set to $\epsilon_0 = \beta = 1$. \bb{Bottom:} For 
 the same parameters we plot an often used canditate for the EP rate over dimensionless time $\Gamma t$ in logarithmic 
 scale for the initial state $p_1(0) = 0.1, p_0(0)= 0.9$. }
\end{figure}

We now investigate the IFP of the dynamics. In Fig.~\ref{fig plot ex 2} (top) we plot the time 
evolution of the rates $\gamma_\text{in}(t)$ and $\gamma_\text{out}(t)$ as well as their ratio. We see that for long 
times they become stationary and their ratio fulfills local detailed balance~(\ref{eq local detailed balance}), which 
implies that the steady state is a Gibbs state and hence, the system properly thermalizes. However, for short times, the 
ratio does not fulfill local detailed balance and hence, the IFP is not the Gibbs state. Furthermore, as the rates are 
positive all the time, the dynamics is clearly 1-Markovian. This proves that a 1-Markovian time-evolution, which yields 
the correct long-time equilibrium state, can nevertheless have a time-dependent IFP, even if the underlying Hamiltonian 
is time-independent. This clearly shows that 1-Markovian evolution does not imply a time-invariant IFP as claimed 
in the literature [see, e.g., below Eq.~(47) in Ref.~\cite{DeVegaAlonsoRMP2017} or Eq.~(9) in 
Ref.~\cite{ThomasEtAlPRE2018}]. 

In addition, Fig.~\ref{fig plot ex 2} (bottom) also shows the time evolution of 
\begin{equation}
 \dot\sigma(t) \equiv -\frac{\partial}{\partial t}D[\hat\rho_S(t)\|e^{-\beta\hat H_S}/Z_S]
\end{equation}
In the weak coupling limit it is tempting to identifiy $\dot\sigma(t)$ as the EP rate because the global equilibrium 
state can be approximated by $\hat\pi_{SB} \approx e^{-\beta\hat H_S}/Z_S\otimes e^{-\beta\hat H_B}/Z_B$. However, one 
should be cautious here as this is not an exact result and the initial product state assumption does not fit into the 
description used in Secs.~\ref{sec classical system bath theory} and~\ref{sec thermo quantum}. The transient dynamics 
is indeed dominated by the build-up of system-bath correlations and an exact treatment needs to take them into 
account~\cite{EspositoLindenbergVandenBroeckNJP2010}. Therefore, outside the specific limit of the Born-Markov secular 
master equation, where $\dot\sigma(t)$ can be related to the actual EP rate~\cite{SpohnLebowitzAdvChemPhys1979, 
SpohnJMP1978}, the quantity $\dot\sigma(t)$ lacks a clear connection to a consistent thermodynamic framework. In addition, 
Fig.~\ref{fig plot ex 2} clearly demonstrates that $\dot\sigma(t) < 0$ is possible although the dynamics is 1-Markovian. 
For these reasons the claimed connections between a negative ``entropy production'' rate $\dot\sigma(t)$ and 
non-Markovianity in Refs.~\cite{ArgentieriEtAlEPL2014, BhattacharyaEtAlPRA1017, MarcantoniEtAlSR2017, 
PopovicVacchiniCampbellPRA2018} require a careful reassessment.

%%%%%%%%%%%%%%%%%%%%%%%%%%%%%%%%%%%%%%%%%%%%%%%%%%%%%%%%%%%%%%%%%%%%%%%%%%%%%%%%%%%%%%%%%%%%%%%%%%%%%%%%%%%%%%%%%%%%%%%%
\section{Summary and outlook}
\label{sec summary and outlook}

%%%%%%%%%%%%%%%%%%%%%%%%%%%%%%%%%%%%%%%%%%%%%%%%%%%%%%%%%%%%%%%%%%%%%%%%%%%%%%%%%%%%%%%%%%%%%%%%%%%%%%%%%%%%%%%%%%%%%%%%
\subsection{Summary}
\label{sec summary}

A large part of this paper was devoted to study the instantaneous thermodynamics at the rate level for an arbitrary 
classical system coupled to a single heat bath. Quite remarkably, the definition of the EP rate~(\ref{eq 2nd law intro}) 
for a weakly coupled Markovian system can be carried over to the strong-coupling and non-Markovian situation if 
we replace the Gibbs state with the correct equilibrium state $\pi_\alpha(\lambda_t)$, described, e.g., by the 
Hamiltonian of mean force~\cite{KirkwoodJCP1935}. The EP rate then reads 
\begin{equation}\label{eq EP rate generalized}
 \dot\Sigma(t) \equiv -\left.\frac{\partial}{\partial t}\right|_{\lambda_t} D[p_\alpha(t)\|\pi_\alpha(\lambda_t)].
\end{equation}
Starting from this definition together with an unambiguous definition for work [Eqs.~(\ref{eq work bipartite}) 
and~(\ref{eq work rate Hamiltonian})], we recovered the previously proposed definitions in 
Refs.~\cite{SeifertPRL2016, MillerAndersPRE2017, StrasbergEspositoPRE2017}. Most importantly, we were able to connect 
the abstract concept of (non-) Markovianity to the physical observable consequence of having a negative EP rate 
$\dot\Sigma(t) < 0$. We can summarize our finding as follows: 

\begin{thm*}
 If the dynamics are undriven ($\dot\lambda_t = 0$), any appearance of $\dot\Sigma(t) < 0$ unambiguously reveals that 
 the dynamics is non-Markovian. If the dynamics is driven ($\dot\lambda_t \neq 0$), any appearance of $\dot\Sigma(t) < 0$ 
 unambiguously reveals that the dynamics is non-Markovian \bb{or} that $\pi_\alpha(\lambda_t)$ cannot be an IFP of the 
 dynamics. This implies that TSS does not apply. 
\end{thm*}

Especially for the undriven case, it was important to study the question when is the equilibrium state 
$\pi_\alpha(\lambda_t)$ also an IFP of the dynamics. To the best of our knowledge, this was not yet studied thoroughly. 
In particular, a 1-Markovian evolution of the system does \emph{not} imply that $\pi_\alpha(\lambda_t)$ is an 
instantaneous fixed point of the dynamics. This is the reason why a 1-Markovian evolution alone is not sufficent to imply 
that the entropy production rate is always positive. Fig.~\ref{fig overview} shows the mathematical implications 
and equivalences worked out in this paper. 

\begin{figure*}%[b]
 \centering\includegraphics[width=0.85\textwidth,clip=true]{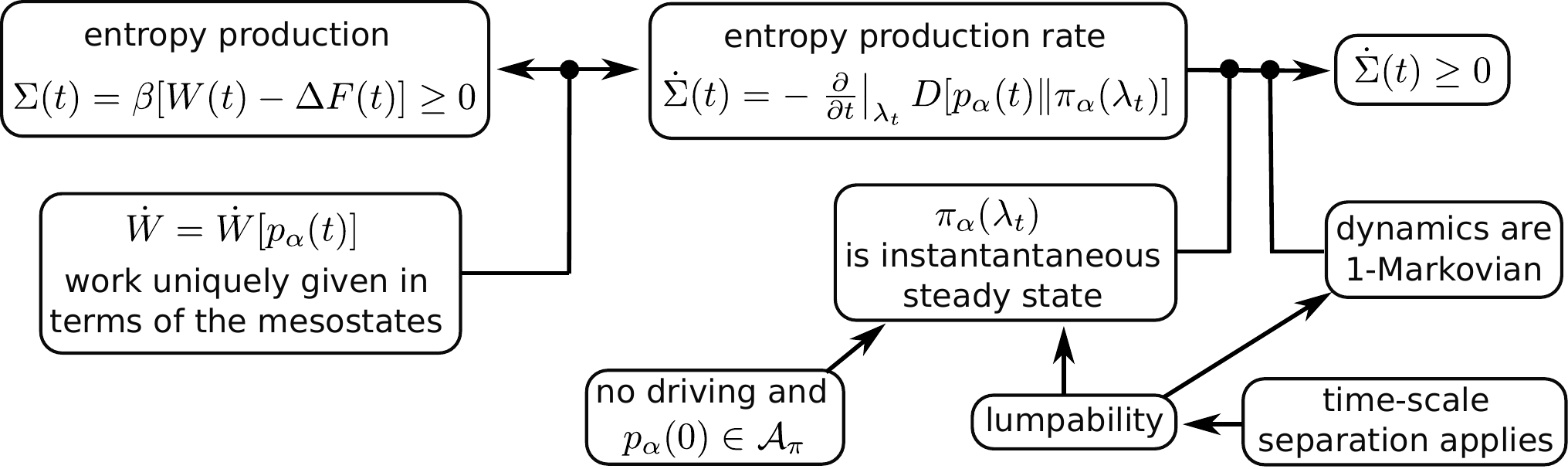}
 \label{fig overview} 
 \caption{Overview of the results from 
 Secs.~\ref{sec mathematical results},~\ref{sec coarse-grained dissipative dynamics} 
 and~\ref{sec classical system bath theory} (the notation is chosen as in Secs.~\ref{sec mathematical results} 
 and~\ref{sec coarse-grained dissipative dynamics}, but the findings are identical to 
 Sec.~\ref{sec classical system bath theory}). The arrows indicate implications in a mathematical sense. 
 Some implications depend on certain conditions, which are marked by a line attached with a circle to the 
 respective arrow. }
\end{figure*}

We then left the classical regime and provided a thermodynamic framework for a strongly coupled, driven \emph{quantum} 
system immersed in an arbitrary heat bath in Sec.~\ref{sec thermo quantum}. Inspired by the classical treatment and 
backed up by equilibrium considerations using the quantum Hamiltonian of mean force~\cite{HaenggiIngoldTalknerNJP2008, 
GelinThossPRE2009, HsiangHuEntropy2018}, we defined internal energy $U$, system entropy $S$ and free energy $F$ [Eqs.~(\ref{eq def U quantum}) to~(\ref{eq def F quantum})] for a quantum system arbitrarily far from equilibrium. 
Remarkably, the basic definitions are formally identical to the classical case albeit they were critically debated in 
Refs.~\cite{HaenggiIngoldTalknerNJP2008,GelinThossPRE2009}. 
Nevertheless, they ensure that the first and second law as known from phenomenological non-equilibrium 
thermodynamics, $\Delta U = Q + W$ and $\Sigma = \beta(W-\Delta F) = \Delta S - \beta Q \ge 0$, also hold in the 
quantum regime. Thus, at the integrated level the quantum nature of the interaction becomes manifest only by 
realizing that we can treat a smaller class of admissible initially correlated states. At the rate level, however, we 
showed that the quantum generalization of Eq.~(\ref{eq EP rate generalized}) does not coincide with the entropy 
production rate $\dot\Sigma(t) = \beta[\dot W(t) - d_t F(t)]$. Thus, at present it seems that there is no rigorous 
connection between negative entropy production rates and non-Markovianity. 

To support the latter statement we also investigated in Sec.~\ref{sec quantum example} what happens for initially 
decorrelated states if we use the conventional definition of entropy production rate [i.e., the quantum counterpart 
of Eq.~(\ref{eq 2nd law intro})] valid in the limit of the Born-Markov-secular approximation~\cite{SpohnJMP1978, 
SpohnLebowitzAdvChemPhys1979, LindbladBook1983, BreuerPetruccioneBook2002, KosloffEntropy2013}. Unfortunately, 
outside this limit this definition does not provide an adequate candidate for an entropy production rate and even for 
a weakly coupled and 1-Markovian system it can be transiently negative. From the perspective of open quantum system 
theory, this behaviour is caused by the initial build-up of system-environment correlations, which -- even in the weak 
coupling limit -- cannot be neglected and need to be taken into account in any formally exact thermodynamic 
framework~\cite{EspositoLindenbergVandenBroeckNJP2010}. 

Table~\ref{table quantum vs classical} summarizes what is known (and what not) about the thermodynamic description 
of a driven system coupled to a single heat bath for the classical (abbreviated CM) and the quantum (QM) case, 
respectively. 

\begin{table}[h]
 \centering
  \begin{tabular}{l|c|c}
   				&	CM   &   QM			\\
   \hline 
   Consistent with equilibrium thermodynamics$^{(a)}$  &   \checked   &    \checked   \\
   Nonequilibrium first law                            &   \checked   &    \checked   \\
   Nonequilibrium second law                           &   \checked   &    \checked   \\
   Recovery of weak-coupling limit                     &   \checked   &    \checked   \\
   Jarzynski-Crooks work fluctuation theorem$^{(b)}$   &   \checked   &    \checked   \\
   Entropy production fluctuation theorem$^{(c)}$      &   \checked   &    \lightning  \\
   Arbitrary initial system states                     &   \checked   &    \lightning  \\
   Consistent with TSS                                 &   \checked   &    \lightning  \\
   Connection to non-Markovianity                      &   \checked   &    \lightning  \\
  \end{tabular}
  \caption{\label{table quantum vs classical} Current state-of-the-art of strong coupling thermodynamics for a 
  single heat bath. The \lightning-symbol indicates only that it is \emph{currently} not known how to establish the 
  corresponding quantum version. Remarks: (a) We here mean that the standard textbook relations between the partition 
  function and internal energy, entropy and free energy are recovered at equilibrium. (b) A work fluctuation theorem 
  of the ``Jarzynski-Crooks'' type starts with a process in equilibrium and contains the \emph{equilibrium} free 
  energies in the expression. (c) An entropy production (or ``integral'') fluctuation theorem allows to start in a 
  non-equilibrium state and contains the \emph{nonequilibrium} free energies. 
  }
\end{table}

%%%%%%%%%%%%%%%%%%%%%%%%%%%%%%%%%%%%%%%%%%%%%%%%%%%%%%%%%%%%%%%%%%%%%%%%%%%%%%%%%%%%%%%%%%%%%%%%%%%%%%%%%%%%%%%%%%%%%%%%
\subsection{Outlook}
\label{sec outlook}

After having established a general theoretical description involving a lot of mathematical details, we here take the 
freedom to be less precise in order to discuss various consequences of our findings and to point out interesting open 
research avenues. 

First of all, the field of strong coupling and non-Markovian thermodynamics is far from being settled and many different 
approaches have been put forward. Therefore, one might wonder whether the definitions we have used here are the 
``correct'' ones or whether one should not start with a completely different set of definitions. We believe that the 
definitions we have used possess a certain structural appeal: we could establish a first and second law as known from 
phenomenological non-equilibrium thermodynamics and in the limit of TSS or at equilibrium, our definitions coincide with 
established results from the literature. Furthermore, the fact that in the classical case we could give to the 
appearance of a negative EP rate a clear dynamical meaning adds further appeal to the definitions used here. 

On the other hand, this last point is lost for quantum systems leaving still a larger room of ambiguity there. 
In this respect, it is also worth to point out that for strongly coupled, non-Markovian systems it was also possible to 
find definitions which guarantee an always positive EP rate even in presence of multiple heat baths. One possibility 
is to redefine the system-bath partition~\cite{StrasbergEspositoPRE2017, StrasbergEtAlNJP2016, NewmanMintertNazirPRE2017, 
SchallerEtAlPRB2018, StrasbergEtAlPRB2018, RestrepoEtAlNJP2018}, which reverses the strategy of 
Sec.~\ref{sec coarse-grained dissipative dynamics}: instead of looking at the mesostates only when starting from a 
consistent description in terms of the microstates, one starts with a mesoscopic description and ends up with a 
consistent description in a larger space, i.e., one effectively finds the microstates from 
Sec.~\ref{sec coarse-grained dissipative dynamics}. Alternatively and without enlarging the state space, Green's 
functions techniques can be used for simple models to define an always positive EP 
rate~\cite{EspositoOchoaGalperinPRL2015, BruchEtAlPRB2016, LudovicoEtAlPRB2016, HaughianEspositoSchmidtPRB2018} or the 
Polaron transformation can be useful when dealing with particular strong coupling situations~\cite{SchallerEtAlNJP2013, 
KrauseEtAlJCP2015, GelbwaserKlimovskyAspuruGuzikJPCL2015, WangRenCaoSciRep2015, FriedmanAgarwallaSegalNJP2018}. 

Applying our present framework in context of \emph{multiple} heat baths poses a formidable challenge as it remains 
unclear what the correct reference state $\pi_\alpha(\lambda_t)$ should be. While it is known how to extend the second 
law~(\ref{eq 2nd law intro}) to multiple heat baths if the Born-Markov secular approximation is 
applied~\cite{SpohnLebowitzAdvChemPhys1979}, this approximation can be unjustified even at weak 
coupling~\cite{MitchisonPlenioNJP2018}. Furthermore, the correct choice of initial state plays a crucial role as it 
can lead to different thermodynamic definitions; compare, e.g., with the initial product state assumption used in 
Ref.~\cite{EspositoLindenbergVandenBroeckNJP2010}. At the end, we believe that the most meaningful thermodynamic 
description will indeed depend on the question which degrees of freedom we can measure and control in an experiment. 
However, at least at steady state many of the different approaches coincide because the system-bath boundary then 
usually contributes only a time-independent additive constant to the description. 

Within the framework we have used here, we can get also more insights by viewing our findings in light of the recent 
endeavour to find a meaningful quantifier of non-Markovianity for quantum systems~\cite{RivasHuelgaPlenioRPP2014, 
BreuerEtAlRMP2016}. At least for classical, undriven systems it seems reasonable to measure the degree of 
non-Markovianity via the quantity 
\begin{equation}\label{eq quantifier NM}
 \C N \equiv \max_{p_x(0)\in\C A_\pi}\int_{\dot\Sigma(t) < 0} \left|\dot\Sigma(t)\right| dt \ge 0.
\end{equation}
The larger $\C N$, the stronger the system behaves non-Markovian. This quantifier shares structural similarity with the 
BLP quantifier~\cite{BreuerLainePiiloPRL2009} and a non-zero value could be likewise interpreted as information backflow 
from the bath to the system. Thus, our findings show that due to memory effects $\dot\Sigma(t)$ looses its property of a 
Lyapunov function. Of course, $\C N$ presents just one out of a 
multitude of possible non-Markovianity quantifiers~\cite{RivasHuelgaPlenioRPP2014, BreuerEtAlRMP2016}, but it has the 
outstanding advantage that it is clearly linked to an important and meaningful physical quantitiy. Its comparison with 
other measures therefore deserves further attention. 

To close this paper, we ask for which problems non-Markovian effects could be beneficial in a thermodynamic sense. 
This question constitutes in principle a vast field on its own, which we only want to briefly touch. A central benefit 
of non-Markovian dynamics is that new state transformations become possible, which are not realizable with a Markovian 
finite time dynamics.\footnote{The question whether a given initial state $p_\alpha(0)$ can be transformed into 
a given final state $p_\alpha(t)$ by a Markovian ME is known as the ``embedding problem''. For a recent account of this 
field see Ref.~\cite{LencastreEtAlPRE2016}. The problem was also studied quantum mechanically in 
Ref.~\cite{WolfEtAlPRL2008}. } We here want to give a simple example of physical and thermodynamic relevance to illustrate 
the main point. This example is the erasure of a single bit of information. 

Erasing a single bit of information is related to Landauer's famous principle~\cite{LandauerIBM1961} and it is nowadays 
possible to measure the minuscule thermodynamic changes associated to this transformation~\cite{OrlovEtAlJJAP2012, 
BerutEtAlNature2012, JunGavrilovBechhoeferPRL2014, BerutPetrosyanCilibertoJSM2015, GavrilovBechhoeferPRL2016, 
HongEtAlSciAdv2016, YanEtAlPRL2018}. Theoretically, the process of erasure is usually modeled with a Markovian two-state 
system and optimal protocols have been investigated in Refs.~\cite{DianaBagciEspositoPRE2013, ZulkowskiDeWeesePRE2014}. 
Let us now illustrate which benefits non-Markovian dynamics can add. We denote the two states of the bit by ``0'' and 
``1'' and model the dynamics by the ME 
\begin{equation}
 \frac{\partial}{\partial t}\binom{p_1(t)}{p_0(t)} = 
 \left(\begin{array}{cc}
        -\gamma_{01}(t)	&	\gamma_{10}(t)	\\
        \gamma_{01}(t)	&	-\gamma_{10}(t)	\\
       \end{array}\right) \binom{p_1(t)}{p_0(t)}.
\end{equation}
Since we have not made any assumptions about the time-dependent rates $\gamma_{01}(t)$ and $\gamma_{10}(t)$, this model 
is general and could be obtained directly from Eq.~(\ref{eq ME meso general}). Note that the origin of the 
time-dependence of the rates does not need to come from any driving, cf. Eqs.~(\ref{eq ME meso general}) 
or~(\ref{eq ME SRL}). From $p_0(t) + p_1(t) = 1$ we obtain a 
linear, inhomogeneous differential equation with time-dependent coefficients for the probability to be in state zero. 
It reads $\dot p_0(t) = \gamma_{01}(t) - [\gamma_{10}(t)+\gamma_{01}(t)] p_0(t)$ with the formal solution 
\begin{align}
 p_0(t)	=&~	\exp\left[-\int_0^t ds [\gamma_{10}(s)+\gamma_{01}(s)]\right] p_0(0)    \label{eq solution bit}	\\
		&+	\int_0^t ds \exp\left[-\int_s^t du [\gamma_{10}(u)+\gamma_{01}(u)]\right] \gamma_{01}(s).  \nonumber
\end{align}
For definiteness we choose to erase the bit such that the probability $p_0(t)$ to find the bit in state zero is as large 
as possible at time $t$. 

Now, as a proof of principle, let us assume that $\gamma_{01}(t) \ge 0$ for all times $t$, but $\gamma_{10}(t)$ can be 
negative for certain times, which clearly indicates non-Markovian behaviour. Furthermore, we denote the fact that 
$p_0(t)$ depends on the whole history of $\gamma_{10}(t)$ by $p_0(t) = p_0[t;\{\gamma_{10}(t)\}]$. Next, we recall the 
well-known inequality $\int_0^t ds f(s) \le \int_0^t ds |f(s)|$ for any time-dependent function $f(t)$, which implies 
\begin{equation}
 \exp\left[-\int_0^t ds f(s)\right] \ge \exp\left[-\int_0^t ds |f(s)|\right].
\end{equation}
Because the two terms in Eq.~(\ref{eq solution bit}) are separately positive, this inequality implies 
\begin{equation}
 p_0[t;\{\gamma_{10}(t)\}] \ge p_0[t;\{|\gamma_{10}(t)|\}]
\end{equation}
for any initial state and independent of the precise form of the rates. In fact, if for certain times 
$\gamma_{10}(t) < 0$ we have a strict inequality: $p_0[t;\{\gamma_{10}(t)\}] > p_0[t;\{|\gamma_{10}(t)|\}]$. 
This shows that non-Markovian effects can help erase a bit faster in finite time. 

To conclude, we believe that our work paves the way for a rigorous understanding of finite-time thermodynamics away 
from the conventional Markovian assumption. Because our understanding of finite-time processes has drastically improved 
during the last years~\cite{DeffnerCampbellJPA2017}, exploring their thermodynamic implications opens up a new 
and exciting research field.

%%%%%%%%%%%%%%%%%%%%%%%%%%%%%%%%%%%%%%%%%%%%%%%%%%%%%%%%%%%%%%%%%%%%%%%%%%%%%%%%%%%%%%%%%%%%%%%%%%%%%%%%%%%%%%%%%%%%%%%%
\subsection*{Acknowledgements}

This research is funded by the European Research Council project NanoThermo (ERC-2015-CoG Agreement No. 681456).

%%%%%%%%%%%%%%%%%%%%%%%%%%%%%%%%%%%%%%%%%%%%%%%%%%%%%%%%%%%%%%%%%%%%%%%%%%%%%%%%%%%%%%%%%%%%%%%%%%%%%%%%%%%%%%%%%%%%%%%%

\bibliography{/home/philipp/Documents/references/books,/home/philipp/Documents/references/open_systems,/home/philipp/Documents/references/thermo,/home/philipp/Documents/references/info_thermo,/home/philipp/Documents/references/general_QM,/home/philipp/Documents/references/math_phys}
%\bibliography{books,open_systems,thermo,general_QM,info_thermo,math_phys}

%%%%%%%%%%%%%%%%%%%%%%%%%%%%%%%%%%%%%%%%%%%%%%%%%%%%%%%%%%%%%%%%%%%%%%%%%%%%%%%%%%%%%%%%%%%%%%%%%%%%%%%%%%%%%%%%%%%%%%%%
\appendix

%%%%%%%%%%%%%%%%%%%%%%%%%%%%%%%%%%%%%%%%%%%%%%%%%%%%%%%%%%%%%%%%%%%%%%%%%%%%%%%%%%%%%%%%%%%%%%%%%%%%%%%%%%%%%%%%%%%%%%%%
\section{Weak lumpability}
\label{sec app weak lumpability}

The notion of lumpability required the coarse-grained Markov chain to be Markovian for any initial microstate. 
One might wonder what can be said about the dynamics if there is at least one initial microstate which leads to a 
Markov chain at the mesolevel. For this purpose Kemeny and Snell introduce the concept of weak 
lumpability (Sec.~6.4. in Ref.~\cite{KemenySnellBook1976}): 

\begin{mydef}[Weak lumpability]\label{def weak lumpability}
 A Markov chain is weakly lumpable with respect to a partition $\boldsymbol\chi$ if there exists at least one initial 
 distribution $p_{x}(0)$ such that the lumped process is a Markov chain. The TM can then depend on $p_{x}(0)$. 
\end{mydef}

The fact that the TMs for a weakly lumpable process can depend on the initial microstate $p_x(0)$ is also apparent in 
Eq.~(\ref{eq TM mesolevel}). Furthermore, it is again clear that a weakly lumpable process with respect to $p_x(0)$ 
for a given TM $T_\tau$ and partition $\boldsymbol\chi$, is also a weakly lumpable process with respect to $p_x(0)$ 
for all larger times, i.e., for all $T_{n\tau} = (T_\tau)^n$ with $n>1$ and the same partition $\boldsymbol\chi$. 

The concept of weak lumpability is especially useful when the underlying Markov chain is regular\footnote{In 
Sec.~\ref{sec steady states} we called this property ergodicity, which is more familiar for a physicist. Here, instead,
we follow the terminology of Ref.~\cite{KemenySnellBook1976}. }: 

\begin{mydef}[Regular Markov chain]
 A Markov chain is called regular if there exists an $n\in\mathbb{N}$ such that all elements of the matrix $T^n_{\tau}$ 
 are strictly positive. 
\end{mydef}

A regular Markov chain ensures that the system reaches its steady state 
$\boldsymbol\pi = \lim_{n\rightarrow\infty} T^n_{\tau} \bb p(0)$ for any initial distribution $\bb p(0)$ and hence, 
it has a unique steady state. Kemeny and Snell then prove the following~\cite{KemenySnellBook1976}: 

\begin{thm}\label{thm weak lumpability}
 Assume that a regular Markov chain with steady state $\pi_x$ is weakly lumpable with respect to the partition 
 $\boldsymbol\chi$ for some initial distribution $p_x(0)$. Then, the Markov chain is also weakly lumpable for the 
 initial distribution $\pi_x$ with the same transition probabilities, which are determined by 
 \begin{equation}\label{eq TM weak lumpability}
  G_{\tau}(\alpha|\beta) = \sum_{x_\alpha,y_\beta} T_{\tau}(x_\alpha|y_\beta) \pi_{y|\beta}.
 \end{equation}
\end{thm}

Thus, Theorem~\ref{thm weak lumpability} says that for a regular and weakly lumpable Markov chain we can always use 
the conditional steady state $\pi_{x|\alpha}$ to construct the TM at the mesolevel and do not need to use 
$p_{x|\alpha}(0)$ as in Eq.~(\ref{eq TM mesolevel}). This is advantageous to say something about the IFP of undriven 
processes: 

\begin{thm}\label{thm steady state weak appendix}
 Consider an undriven stochastic process described by the ME~(\ref{eq ME meso general}), i.e., we assume $G_{t,0}^{-1}$ 
 to exist for all admissible initial states $\C A(0)$ and all times $t$. If the stochastic process is weakly lumpable 
 for an underlying regular Markov chain with respect to an admissible initial state $p_x(0)\in\C A(0)$, 
 then $\pi_\alpha$ is a IFP of the stochastic process at the mesolevel. 
\end{thm}

\begin{proof} 
 Using the insights from Theorem~\ref{thm weak lumpability}, it becomes clear that 
 $\sum_\beta G_{t,0}(\alpha|\beta)\pi_\beta = \pi_\alpha$ for all times $t$. Together with the invertibility condition 
 we also get Eq.~(\ref{eq invertible steady state}) from the main text. These two relations were all we needed to ensure 
 that Eq.~(\ref{eq help thm steady state weak}) holds. 
\end{proof}

To conclude, as most physically relevant Markov chains are regular, the concept of weak lumpability helps us to deal 
with initial conditions, where the microstates have not reached a conditional steady state. Together with 
Theorem~\ref{thm ent prod} this would then imply an always positive EP rate because a weakly lumpable process with 
respect to an admissible initial state $p_x(0)\in\C A(0)$ is also 1-Markovian with respect to that state. 

However, a weakly lumpable process still requires the whole hierarchy to fulfill the Markov 
condition~(\ref{eq cond Markovianity}) and also on physical grounds we expect that it is a good approximation to 
assume that the conditional initial microstates are at steady state. If this is the case, then the notion of weak 
lumpability does not seem to add any further insights into the theory of 
Secs.~\ref{sec coarse-grained dissipative dynamics} and~\ref{sec classical system bath theory}.

%%%%%%%%%%%%%%%%%%%%%%%%%%%%%%%%%%%%%%%%%%%%%%%%%%%%%%%%%%%%%%%%%%%%%%%%%%%%%%%%%%%%%%%%%%%%%%%%%%%%%%%%%%%%%%%%%%%%%%%%
\section{Instantaneous fixed points and time-local master equation}
\label{sec app IFP}

Formally exact time-local ME can be derived in different ways. One particular construction was given in 
Eq.~(\ref{eq meso generator ME}), but another possibility is given by the time-convolutionless 
ME~\cite{BreuerPetruccioneBook2002, DeVegaAlonsoRMP2017, FulinskiKramarczyk1968, ShibataTakahashiHashitsumeJSP1977} 
and see Ref.~\cite{AnderssonCresserHallJMO2007} for yet another way of construction. We will here show that, as long as 
the inverse of the TM $G_{t,0}$ defined in Eq.~(\ref{eq TM mesolevel general}) exists, the generators all coincide. 
Hence, the IFP computed with any of those time-local MEs is the same and therefore the IFP is a well-defined concept. 

To see this, let us denote by $V^{(1)}(t)$ and $V^{(2)}(t)$ the generators of an exact time-local ME derived in two 
different ways (we suppress the depencence on $\lambda_t$ here for simplicity). Because both are assumed to be formally 
exact for any admissible initial condition, we have 
\begin{equation}
 \sum_\beta[V_{\alpha,\beta}^{(1)}(t)-V_{\alpha,\beta}^{(2)}(t)] p_\beta(t) = 0
\end{equation}
for any mesostate $p_\beta(t)$, which is reachable from the class of admissible initial states $\C A(0)$. This equation 
also holds for any linear combination of such states, i.e., 
\begin{equation}
 \sum_i \mu_i \sum_\beta[V_{\alpha,\beta}^{(1)}(t)-V_{\alpha,\beta}^{(2)}(t)] p^{(i)}_\beta(t) = 0
\end{equation}
with $\mu_i\in\mathbb{R}$. We now use that $G_{t,0}$ is invertible for any finite $t$, which implies in particular that 
the dimension of the image of $G_{t,0}$ cannot decrease. But since the class of admissible initial states spans the 
entire vector space including all probability distributions $p_\beta(t)$, we can always choose 
\begin{equation}
 \sum_i \mu_i p^{(i)}_\beta(t) = \delta_{\beta,\beta'}
\end{equation}
for any $\beta'$. This implies that $V_{\alpha,\beta'}^{(1)}(t)-V_{\alpha,\beta'}^{(2)}(t) = 0$ 
for any $\alpha$ and $\beta'$. Hence, $V^{(1)}(t) = V^{(2)}(t)$.

%%%%%%%%%%%%%%%%%%%%%%%%%%%%%%%%%%%%%%%%%%%%%%%%%%%%%%%%%%%%%%%%%%%%%%%%%%%%%%%%%%%%%%%%%%%%%%%%%%%%%%%%%%%%%%%%%%%%%%%%
\section{Fixed points of coarse-grained Hamiltonian dynamics}
\label{sec app Hamiltonian dynamics}

In this appendix we rederive those results from Sec.~\ref{sec mathematical results}, which will be of relevance for 
Sec.~\ref{sec classical system bath theory}. Let us start with an arbitrary Hamiltonian $H(\lambda_t)$ and an arbitrary 
fixed partition $\boldsymbol\chi$. The ``master equation'' corresponding to this Hamiltonian is the Liouville equation 
\begin{equation}\label{eq Liouville equation}
 \frac{\partial}{\partial t}\rho(x;t) = \{H(x;\lambda_t),\rho(x;t)\},
\end{equation}
where $\{\cdot,\cdot\}$ denotes the Poisson bracket and $\rho(x;t)$ is the probability distribution defined on the 
phase space consisting of the collection of all positions $\bb q$ and momenta $\bb p$. For simplicity and analogy with 
the main text we denote a point in phase space by $x = (\bb q,\bb p)$. 
For any partition $\boldsymbol\chi$, the mesostates are defined as 
\begin{equation}
 \rho(\alpha;t) = \int_{\chi_\alpha} dx \rho(x;t),
\end{equation}
where $\alpha$ can be continuous (e.g., if we trace out a bath) or discrete (e.g., if we lump the motion of a particle 
in a double well potential into two states ``left'' and ``right''). Furthermore, it turns out to be convenient to denote 
the dynamical map generated by Eq.~(\ref{eq Liouville equation}) over a finite time interval by $\Phi_{t,0}$, i.e., 
\begin{equation}
 \rho(x;t) = \int dx' \Phi_{t,0}(x|x')\rho(x';0),
\end{equation}
similar to the time-evolution operator in quantum mechanics. 

Clearly, as in Sec.~\ref{sec mathematical results} for a given conditional initial microstate $\rho(x|\alpha;0)$, 
$\Phi_{t,0}$ induces a map at the mesolevel, 
\begin{align}
 \rho(\alpha;t)			&=	\int d\beta \C G_{t,0}(\alpha|\beta)\rho(\beta;0),	\\
 \C G_{t,0}(\alpha|\beta)	&\equiv \int_{\chi_\alpha} dx \int_{\chi_\beta} dx' \Phi_{t,0}(x|x') \rho(x'|\beta;0).
\end{align}
Using the procedure outlined in Sec.~\ref{sec time dependent MEs} or the time-convolutionless 
ME~\cite{FulinskiKramarczyk1968, ShibataTakahashiHashitsumeJSP1977, BreuerPetruccioneBook2002, DeVegaAlonsoRMP2017}, 
we write the time evolution of the mesostate again in terms of a formally exact ME [cf.~Eq.~(\ref{eq ME meso general})] 
\begin{equation}\label{eq ME continuous}
 \frac{d}{dt}\rho(\alpha;t) = \C V(\lambda_t,t)\rho(\alpha;t).
\end{equation}

We note that the Hamiltonian dynamics generated by $\Phi_{t,s}$ ($t\ge s$) are Markovian. It is therefore possible to 
straightforwardly extend the definition of lumpability to Hamiltonian dynamics for any propagator $\Phi_{t,s}$. 
We will here choose a version for the infinitesimal propagator $\Phi_{t+\delta t,t}$ which is most useful in the 
following. 

\begin{mydef}[Lumpability -- continuous version]\label{def strong lumpability cont}
 The dynamics generated by Eq.~(\ref{eq Liouville equation}) is lumpable with respect to the partition 
 $\boldsymbol\chi$ if for every initial distribution $\rho(x;0)$ the lumped process is Markovian and the generator 
 $\C V(\lambda_t,t)$ in Eq.~(\ref{eq ME continuous}) does not depend on $\rho(x;0)$. 
\end{mydef}

We now formulate the analogue of Theorem~\ref{thm steady state strong}: 

\begin{thm}
 If the stochastic process is lumpable as in Definition~\ref{def strong lumpability cont} for some 
 time-interval $I$, then the IFP of the system is given by the marginal global Gibbs state 
 with respect to $H(x;\lambda_t)$ for all $t\in I$, i.e., $\C V(\lambda_t,t)\pi(\alpha;\lambda_t) = 0$ with 
 \begin{equation}
  \pi(\alpha;\lambda_t) = \int_{\chi_\alpha} dx \frac{e^{-\beta H(x;\lambda_t)}}{Z(\lambda_t)}.
 \end{equation}
\end{thm}

\begin{proof}
 By assumption, the dynamics of the system is generated by 
 \begin{equation}
  \begin{split}
   \frac{d}{dt}\pi(\alpha;\lambda_t)	&=	\C V(\lambda_t,t)\pi(\alpha;\lambda_t)	\\
					&=	\int_{\chi_\alpha} dx\{H(x;\lambda_t),\pi(\alpha;\lambda_t)\rho(x|\alpha;t)\},
  \end{split}
 \end{equation}
 where $\rho(x|\alpha;t)$ is so far an unknown conditional microstate and where we used that the reduced dynamics from 
 Eq.~(\ref{eq ME continuous}) is formally exact and thus, they coincide with the coarse-grained global dynamics. 
 
 Next, by assumption of lumpability, we know that $\C V(\lambda_t,t)$ is the same \emph{for any initial state}. 
 Let us choose the particular initial state 
 \begin{equation}
  \rho(x;0) = \Phi_{t,0}^{-1}\pi(x;\lambda_t),
 \end{equation}
 which is obtained by evolving the Gibbs state $\pi(x;\lambda_t)$ at time $t$ backward in time. Since the global 
 dynamics is Hamiltonian, we remark that the inverse of $\Phi_{t,0}$ exists and maps well-defined probability 
 distribution onto well-defined probability distributions. But for this choice we clearly have 
 \begin{equation}
  \frac{d}{dt}\pi(\alpha;\lambda_t) = \int_{\chi_\alpha} dx\{H(x;\lambda_t),\pi(x;\lambda_t)\} = 0.
 \end{equation}
 This implies the theorem. 
\end{proof}

In principle, of course, we expect the concept of lumpability to be of limited use for Hamiltonian dynamics. 
However, we can also establish the first part of Theorem~\ref{thm steady state weak} for Hamiltonian dynamics: 

\begin{thm}
 Consider an undriven Hamiltonian. If the set of admissible initial states obeys 
 $\C A(0) \subset\C A_\pi$, then $\pi(\alpha)$ is a IFP of the stochastic process at the mesolevel. 
\end{thm}

The proof is identical to the first part of the proof of Theorem~\ref{thm steady state weak}. Furthermore, as 
in Theorem~\ref{thm steady state general} this also holds for the time-dependent case whenever 
$\C A(t)\subset\C A_\pi(\lambda_t)$ or $\pi(x;\lambda_t) \in \C A(t)$. 

Finally, Theorem~\ref{thm ent prod} then follows analogously by replacing the discrete relative entropy by its 
differential version and by noting that a proper generalization of Lemma~\ref{lemma Markov contractivity} holds also for 
infinite dimensions~\cite{SpohnJMP1978}.

\end{document}